%% file: sparsity_invariance_TNCS_SINGLECOLUMN.tex
\documentclass[10pt]{article}
\usepackage[margin=1in]{geometry}

\usepackage{setspace}
\usepackage{graphicx} 
\usepackage[justification=centering]{caption}
\usepackage{xcolor} 
\usepackage{cite} 
\usepackage{amsmath}
\usepackage[outdir=./]{epstopdf}

	\usepackage{amssymb}
	 \usepackage{enumerate}
	 \usepackage{balance}
	 \usepackage{empheq}
	 \usepackage{mathtools}
	 \usepackage{multicol}
	 \usepackage{algorithm}
	 \usepackage{cuted}
\usepackage{algpseudocode}
\usepackage{tabularx}

\usepackage{booktabs}
\usepackage{multirow}

\mathchardef\Re="023C
\mathchardef\Im="023D

\DeclareMathOperator*{\minimize}{minimize}

\DeclareMathOperator*{\st}{subject~to}

\newtheorem{theorem}{Theorem}
\newtheorem{lemma}{Lemma}

\newtheorem{corollary}{Corollary}

\newtheorem{definition}{Definition} 
\newtheorem{proof}{Proof}

\newtheorem{remark}{Remark}

\newtheorem{example}{Example}
 
\begin{document}

\title{
	 \LARGE \bf Sparsity Invariance for Convex Design of\\ Distributed Controllers
	}
	 \author{Luca Furieri, Yang Zheng, Antonis Papachristodoulou,  and Maryam Kamgarpour\footnote{This research was gratefully funded by the European Union ERC Starting Grant CONENE.  Antonis Papachristodoulou  
was supported in part by the EPSRC project EP/M002454/1. 
Luca Furieri and Maryam Kamgarpour are with the Automatic Control Laboratory, Department of Information Technology and Electrical Engineering, ETH Z\"{u}rich, Switzerland. E-mails: {\tt\footnotesize \{furieril, mkamgar\}@control.ee.ethz.ch}. 
Yang Zheng is with the School Of Engineering And Applied Sciences, Harvard Center for Green Buildings and Cities, Harvard University. E-mail: {\tt \footnotesize zhengy@g.harvard.edu}.
 Antonis Papachristodoulou is with the Department of Engineering Science, University of Oxford, United Kingdom. Email: {\tt\footnotesize antonis@eng.ox.ac.uk}.
	%
}
}

\maketitle

\begin{abstract} 
We address the problem of designing optimal linear time-invariant (LTI)  sparse controllers for  LTI systems, which corresponds to minimizing a norm of the closed-loop system subject to sparsity constraints on the controller structure. This problem is NP-hard in general and motivates the development of tractable approximations. We characterize a class of convex restrictions based on a new notion of Sparsity Invariance (SI). The underlying idea of SI is to design sparsity patterns for transfer matrices $\mathbf{Y}(s)$ and $\mathbf{X}(s)$ such that any corresponding controller $\mathbf{K}(s)=\mathbf{Y}(s)\mathbf{X}(s)^{-1}$ exhibits the desired sparsity pattern. For sparsity constraints, the approach of SI goes beyond the notion of Quadratic Invariance (QI): 1) the SI approach always yields a convex restriction; 2) the solution via the SI approach is guaranteed to be globally optimal when QI holds and performs at least as well as considering a nearest QI subset. Moreover, the notion of SI naturally applies to designing structured static controllers, while QI is not utilizable. Numerical examples show that even for non-QI cases, SI can recover solutions that are 1) globally optimal and 2) strictly more performing than previous methods.

\end{abstract}

\section{Introduction}
The safe and efficient operation of several large-scale systems, such as the smart grid \cite{dorfler2014sparsity}, biological networks \cite{prescott2014layered}, and automated highways \cite{zheng2017distributed}, relies on the decision making of multiple interacting agents. Coordinating the decisions of these agents is challenged by a lack of complete information of the systems' internal variables. Such limited information arises due to privacy concerns, geographic distance or the challenges of implementing a reliable communication network. 

The celebrated work~\cite{Witsenhausen} 
highlighted that lacking full information can enormously complicate
the design of optimal control inputs. Indeed,
the optimal feedback control policies may not even be linear for the Linear Quadratic Gaussian (LQG) control problem without full output information. The intractability inherent to lack of
full information was investigated in the works \cite{Survey, papadimitriou1986intractable}. The core challenges discussed therein
motivated identifying special cases of optimal control problems with partial information for which
efficient algorithms can be used.

Optimally controlling a linear time-invariant system (LTI) with distributed sensor measurements amounts to computing a linear controller that has a desired sparsity pattern and minimizes a norm of the closed-loop system. For this generally intractable problem, the notion of Quadratic Invariance (QI) was shown to be sufficient \cite{rotkowitz2006characterization} and necessary \cite{QIconvexity} for an exact convex reformulation. A related problem of sensor-actuator architecture co-design  was addressed in \cite{matni2015communication,matni2016regularization} by exploiting QI and using sparsity-inducing norm penalties.
\subsection{Previous work on non-QI cases} Given the importance and intricacy of computing optimal distributed controllers, a variety of approximation methods have been proposed for general systems and information structures that are not QI.  For example, the authors in \cite{SDP} developed semidefinite programs 
that are relaxations of this generally NP-hard problem. However, these relaxations might fail to recover a sparse controller that is stabilizing, as confirmed experimentally in \cite{wang2018convex}. To address this issue, polynomial optimization has been used in \cite{wang2018convex} to obtain a
sequence of convex relaxations which converges to a stabilizing distributed controller. Nevertheless, performance of the recovered solution is not directly addressed in \cite{wang2018convex}. For the finite-horizon control problem, the authors in \cite{dvijotham2015convex} derived convex upper bounds to the non-convex cost function to obtain conservative feasible solutions. However, the theoretical sub-optimality bounds were shown to be loose.  Alternatively, the system level approach \cite{wang2019system} proposed an implementation where controllers are required to share locally estimated disturbances in the state-feedback case and internal controller states in the output-feedback case. We note that the classical distributed control only requires to share output measurements, but no intermediate computations, among subsystems. The need to share this additional information in~\cite{wang2019system}  might raise concerns of system security and vulnerability in safety critical applications \cite{sridhar2012cyber}, where each subsystem can only rely on its own sensor measurements.


A different approach to sparse output-feedback controller synthesis is to develop a \emph{convex restriction}: the unstructured problem is reformulated as an equivalent convex program and convex constraints are added to guarantee the desired sparsity pattern of the recovered controllers. Convex restrictions exhibit specific advantages: 1) their optimal solutions can be readily computed with standard convex optimization techniques, and 2) all their feasible solutions are structured and stabilizing by design. A disadvantage is that a restriction may be infeasible even when the original problem is feasible. This motivates developing convex restrictions that are as tight as possible for improved feasibility and performance.  In the literature, convex restrictions have mostly been developed for the special case of computing static controllers \cite{conte2012distributed,zheng2017convex, furieri2019separable}. Within this setting, the problem of optimal sensor and actuator selection was addressed in \cite{lin2013design,dhingra2014admm} with an ADMM approach. For the general case of dynamic controllers given non-QI information structures, the work \cite{rotkowitz2012nearest} suggested restricting the desired sparsity pattern to a subset that is QI to obtain upper bounds on the minimum cost. However, to the best of the authors' knowledge, a method for convex restrictions that can outperform \cite{rotkowitz2012nearest} and goes beyond the notion of QI for sparsity constraints is not known.

\subsection{Contributions}
This paper proposes a generalized framework for the convex design of optimal and near-optimal LTI dynamic output-feedback controllers with a pre-determined sparsity pattern. Our underlying idea is to identify appropriate sparsity patterns for two transfer matrices $\mathbf{Y}(s)$ and $\mathbf{X}(s)$ such that any corresponding feedback controller in the form $\mathbf{K}(s)=\mathbf{Y}(s)\mathbf{X}(s)^{-1}$ exhibits the desired structure. This fundamental property is denoted as Sparsity Invariance (SI). 

Our first contribution is to develop  algebraic conditions on the binary matrices associated with the sparsities of $\mathbf{Y}(s)$ and $\mathbf{X}(s)$ that are necessary and sufficient for SI. Among all such sparsities, we suggest a polynomial-time algorithm to design sparsities that lead to better performance for the distributed control problem at hand. Second, we show that the SI notion steps beyond that of QI in several ways. Indeed,  SI can be applied to general systems subject to arbitrary sparsity constraints, regardless of whether QI holds. Furthermore, SI recovers a controller that is provably globally optimal when QI holds and performs at least as well as that obtained by considering a nearest  QI sparsity subset \cite{rotkowitz2012nearest} when QI does not hold.  Third, we provide examples to show that, even if QI does not hold, controllers obtained through the SI approach can be 1) globally optimal and 2) in general strictly more performing than those obtained using the nearest QI subset approach of \cite{rotkowitz2012nearest}. Finally, we remark that the SI concept is  
applicable to distributed static controller design, as studied in our preliminary work \cite{furieri2019separable}, whereas the Youla parametrization and thus the QI notion is not utilizable. For brevity, our theoretical discussion  focuses on continuous-time systems, but our results also naturally hold for discrete-time systems with sparsity constraints, as we will discuss in the numerical results.  

The rest of this paper is structured as follows. Section~\ref{se:preliminaries} states necessary background and presents the problem formulation. 
Section~\ref{se:sparsity} introduces the class of convex restrictions under investigation and fully characterizes our notion of Sparsity Invariance (SI). We describe how SI can be utilized in an optimized way. In Section~\ref{se:SparsityInvarianceQI}, we show that 1) SI encompasses the previous approaches based on the QI notion, and 2) that strictly better performing sparse controllers can be computed efficiently with the SI approach. We present numerical results in Section~\ref{se:numerical} and conclude the paper in Section~\ref{se:conclusion}.

\section{Background and Problem Statement}
\label{se:preliminaries}

Here, we first introduce some notation on sparsity structures and transfer functions. Then, we state the problem of distributed optimal control, and introduce the necessary background on the Youla parametrization of internally stabilizing controllers.

\subsection{Notation and sparsity structures}	
We use $\mathbb{R},\, \mathbb{C}$ and $\mathbb{N}$ to denote real numbers, complex numbers and positive integers, respectively. The $(i,j)$-th element in a matrix $Y \in \mathbb{R}^{m \times n}$ is referred to as $Y_{ij}$. We use $I_n$ to denote the identity matrix of size $n \times n$,  $0_{m \times n}$ to denote the zero matrix of size $m \times n$ and $1_{m \times n}$ to denote the matrix of size $m \times n$ with all entries set to $1$. 

\emph{Transfer functions:}
    		We denote the imaginary axis as $
    j\mathbb{R}:=\{z \in \mathbb{C} \mid \Re(z)=0\}$ and consider continuous-time transfer functions $\mathbf{f}:j\mathbb{R} \rightarrow \mathbb{C}$. A $m \times n$ \emph{transfer matrix} is the set of $m \times n$ matrices whose entries are transfer functions. We denote the set of $m \times n$ causal transfer matrices as $\mathcal{R}_c^{m\times n}$. A transfer function is called {\emph{proper} (resp. \emph{strictly-proper})} if it is rational and the degree of the numerator polynomial does not exceed {(resp. is strictly lower than)} the degree of the denominator polynomial.  Similar to~\cite{rotkowitz2006characterization}, we denote by $\mathcal{R}_{sp}^{m \times n}$ the set of $m \times n$ strictly proper transfer matrices. Finally, we let $\mathcal{RH}_\infty^{m \times n}$ be the set of $m \times n$ causal and stable transfer matrices.
	
	
	Sparsity structures of transfer matrices can be conveniently represented by binary matrices. A binary matrix is a matrix with entries from the set $\{0,1\}$, and we use $\{0,1\}^{m \times n}$ to denote the set of $m \times n$ binary matrices.  
Given a binary matrix $X \in \{0,1\}^{m \times n}$, we define the associated \emph{sparsity subspace} of causal transfer matrices as
	\begin{align*}
\text{Sparse}(X):=\{&\mathbf{Y}\in  \mathcal{R}_c^{m \times n} \mid  \mathbf{Y}_{ij}(j\omega )=0 ~~\text{for all } i,j \\
&~\text{such that  } X_{ij}=0 \; \text{for almost all $\omega \in \mathbb{R}$} \}\,.
	\end{align*}
 Similarly, given a transfer function $\mathbf{Y} \in \mathcal{R}_c^{m \times n}$, we define $X = \text{Struct}(\mathbf{Y})$ as the binary matrix given by
$$
    X_{ij}:=\begin{cases}
	   0 & \text{if}\; \mathbf{Y}_{ij}(j \omega) = 0\; \text{for almost all}\; \omega \in \mathbb{R},\\
	   1 & \text{otherwise}\,.
	\end{cases}
$$
We say that the transfer matrix $\mathbf{X} \in \mathcal{R}_c^{n \times n}$ is invertible if $\mathbf{X}(j \omega) \in \mathbb{C}^{n \times n}$ is invertible for almost all $\omega \in \mathbb{R}$.

  Let $X, \hat{X} \in \{0,1\}^{m \times n}$ and $Z \in \{0,1\}^{n \times p}$ 
be binary matrices. Throughout the paper, we adopt the following conventions: $X + \hat{X} := \text{Struct}(X + \hat{X})$, and $ XZ:=\text{Struct}(XZ)$.
We say $X \leq \hat{X}$ if and only if $X_{ij}\leq \hat{X}_{ij}\;\forall i,j$, and $X < \hat{X}$ if and only if $X \leq \hat{X}$ and there exist indices $i,j$ such that $X_{ij}<\hat{X}_{ij}$. {Also, we denote} $X \nleq \hat{X}$ if and only if there exist indices $i,j$ such that $X_{ij}>\hat{X}_{ij}$. Given a binary matrix $ X \in \{0,1\}^{m \times n}$ we denote its cardinality, \emph{i.e.}, the total number of nonzero entries, as
	 \begin{equation*}
	  \|X\|_0:=\sum_{i=1}^m \sum_{j=1}^n X_{ij}. 
	 \end{equation*}
Considering the following binary matrices 
\begin{equation*}
        X_1 = \begin{bmatrix}
        0 & 1 & 0 \\ 1 & 1 & 1
    \end{bmatrix}\,, \quad {X}_2 = \begin{bmatrix}
        0 & 1 & 0 \\ 1 & 0 & 1
    \end{bmatrix}\,, \quad {X}_3 = \begin{bmatrix}
        1 & 1 & 0 \\ 1 & 0 & 1
    \end{bmatrix}\,,
\end{equation*} 
we have ${X}_2 < X_1, {X}_3 \nleq X_1$ and ${X}_2 + X_1 = X_1$. Their cardinalities are $\|X_1\|_0 = 4, \|X_2\|_0 = 3$ and $ \|X_3\|_0 = 4$, respectively. For the following transfer matrix,
$$
    \mathbf{Y} = \begin{bmatrix} 0 & \frac{1}{s + 1} & 0 \\ \frac{1}{s + 1} & \frac{1}{s + 1} & \frac{1}{s + 1} \end{bmatrix} \in \mathcal{R}\mathcal{H}_\infty^{2 \times 3}\,,
$$
 if we consider the binary matrix $X_1$ in the example above, we have $\mathbf{Y} \in \text{Sparse}(X_1)$ and $X_1 = \text{Struct}(\mathbf{Y})$.


	\subsection{Problem statement}\label{se:problems}

	
	
	We consider LTI systems in continuous-time
	\begin{align}
	\label{eq:sys_cont}
	&\dot{x}(t)=Ax(t)+Bu(t)+ H_xw(t)\,,\\
	&y(t)=C_yx(t)+H_yw(t) \nonumber\,,\\
	&z(t)=C_zx(t)+D_zu(t)+H_zw(t) \nonumber\,,
	\end{align}
	where $x(t) \in \mathbb{R}^n$, $u(t) \in \mathbb{R}^m$, $y(t) \in \mathbb{R}^{p}$,  $z(t) \in \mathbb{R}^{q}$, and $w(t) \in \mathbb{R}^r$ are the state, control input, observed output, a performance signal defined based on our control objectives, and additive disturbance at time $t \in \mathbb{R}$, respectively. The  input-output transfer function representation {for~\eqref{eq:sys_cont} can be  written as }
\begin{equation*}
\begin{bmatrix}
z\\y
\end{bmatrix}=\mathbf{P}\begin{bmatrix}
w\\u
\end{bmatrix}=\begin{bmatrix}
\mathbf{P}_{11}&\mathbf{P}_{12}\\\mathbf{P}_{21}&\mathbf{G}
\end{bmatrix}\begin{bmatrix}
w\\u
\end{bmatrix}\,,
\end{equation*}
with
$$
    \begin{aligned}
        \mathbf{P}_{11}&:=C_z(s I_n-A)^{-1}H_x+H_z, \\
         \mathbf{P}_{12} &:=C_z(sI_n-A)^{-1}B +D_z, \\
         \mathbf{P}_{21} &:=C_y(s I_n-A)^{-1}H_x+H_y, \\
         \mathbf{G}&:=C_y(s I_n-A)^{-1}B,
    \end{aligned}
$$
where $s$ belongs to $j \mathbb{R}$. Notice that $\mathbf{P}_{11},\mathbf{P}_{12},\mathbf{P}_{21}$ are  proper transfer functions and $\mathbf{G}$ is strictly proper. 

Consider the interconnection of Figure~\ref{fig:interconnection}. 
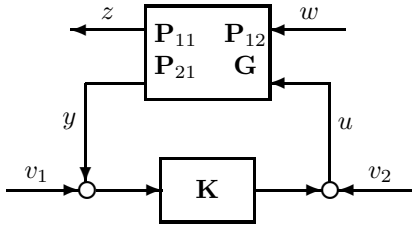
\begin{figure}[t]
\begin{center}\input{figs/f14_2}\end{center}
\caption{\footnotesize Interconnection of $\mathbf{P}$ and $\mathbf{K}$.}
\label{fig:interconnection}
\end{figure}
A dynamic output-feedback controller $u=\mathbf{K}y$ with $\mathbf{K} \in \mathcal{R}_c^{m \times p}$ is said to be \emph{internally stabilizing} if and only if the nine transfer matrices from $w,~\nu_1,~\nu_2$ to $z,~y,~u$ are stable. We denote the set of all causal LTI internally stabilizing output-feedback controllers as $\mathcal{C}_{\text{stab}}$. We say that $\mathbf{P}$ is stabilizable if only and if $\mathcal{C}_{\text{stab}}\neq \emptyset$ and any $\mathbf{K} \in \mathcal{C}_{\text{stab}}$ stabilizes $\mathbf{P}$. Furthermore, we say that a controller $\mathbf{K}$ stabilizes $\mathbf{G}$  if and only if the four transfer matrices from $\nu_1,~\nu_2$ to $y,~u$ are all stable. For the rest of the paper we make the following assumption.

\emph{Assumption 1:} The system $\mathbf{P}$ is stabilizable.

A test for stabilizability of $\mathbf{P}$ is offered in \cite[Chapter 4]{francis1987course}. It is well-known \cite[Chapter 4]{francis1987course}, \cite{rotkowitz2006characterization} that under Assumption~1 a controller $\mathbf{K}$ stabilizes $\mathbf{P}$ if and only if it stabilizes $\mathbf{G}$.  The control problem is to compute a dynamic output-feedback controller $\mathbf{K} \in \mathcal{C}_{\text{stab}}$ which minimizes a given norm $\|\cdot\|$ of
\begin{equation}
f(\mathbf{K})=\mathbf{P}_{11} +\mathbf{P}_{12}\mathbf{K}(I_p-\mathbf{GK})^{-1}\mathbf{P}_{21}\,, \label{eq:closed_loop_cost}
\end{equation}
which is the closed-loop transfer function from $w$ to $z$.

 In distributed control, it is common to add the requirement that $\mathbf{K}$ only uses partial output measurements. This requirement can be captured by adding the constraint $\mathbf{K} \in \text{Sparse}(S)$ for a given binary matrix $S \in \{0,1\}^{m \times p}$, where $S_{ij}=0$ encodes the fact that the $i$-th scalar control input cannot measure the $j$-th measurement output. We formulate this distributed, sparsity-constrained control problem as follows~\cite{rotkowitz2006characterization}:
	\begin{empheq}[box=\fbox]{alignat=3}
	   & \textbf{Problem }\mathcal{P}_K    \nonumber\\
	   & \minimize_{\mathbf{K} \in \mathcal{C}_{\text{stab}} } \qquad && \|f(\mathbf{K})\| \nonumber\\
	   & \st &&\mathbf{K} \in  \text{Sparse}(S)\,, \nonumber
	\end{empheq}	
where $\|\cdot\|$ is any norm of interest. It was shown that a necessary and sufficient condition for a feasible solution to $\mathcal{P}_K$ to exist is that all the distributed fixed modes associated with $S$ lie in the left half of the complex plane \cite{alavian2014stabilizing}. Even if $\mathcal{P}_K$ is feasible, directly computing its optimal solution is intractable because the set $\mathcal{C}_{\text{stab}}$ is non-convex in general. This can be easily verified by checking that when $\mathbf{K}_1,\mathbf{K}_2 \in \mathcal{C}_{\text{stab}}$, the controller $\mathbf{K}=\frac{1}{2}(\mathbf{K}_1+\mathbf{K}_2)$ does not lie in $\mathcal{C}_{\text{stab}}$ in general. Furthermore, the cost function $\|f(\mathbf{K})\|$ is non-convex in $\mathbf{K}$. 

	\subsection{The Youla parametrization of stabilizing controller}
 	 The first step  to convexify problem $\mathcal{P}_K$  is to derive a convex formulation of  the set $\mathcal{C}_{\text{stab}}$ and the function $f(\mathbf{K})$. This is achieved by using a \emph{doubly coprime factorization} of $\mathbf{G}$. 

\begin{lemma}[Chapter~4 of \cite{francis1987course}]
\label{le:doubly}
For any  $\mathbf{G} \in \mathcal{R}_{sp}^{p \times m}$, there exist eight proper and stable transfer matrices defining a doubly coprime factorization of $\mathbf{G}$, that is, they satisfy
\begin{align}
&\mathbf{G}=\mathbf{N}_r\mathbf{M}_r^{-1}=\mathbf{M}_l^{-1}\mathbf{N}_l\,,\nonumber\\
&\begin{bmatrix}
\mathbf{U}_l&-\mathbf{V}_l\\
-\mathbf{N}_l&\mathbf{M}_l
\end{bmatrix}\begin{bmatrix}
\mathbf{M}_r&\mathbf{V}_r\\
\mathbf{N}_r&\mathbf{U}_r
\end{bmatrix}=I_{m+p}\,.\label{eq:dp2}
\end{align}
\end{lemma}
\vspace{2mm}

Then, the Youla parametrization  of all internally stabilizing controllers \cite{youla1976modern} establishes the following equivalence \cite[Chapter~4]{francis1987course}:
\begin{equation}
\label{eq:Cstab}
\mathcal{C}_{\text{stab}}=\{(\mathbf{V}_r-\mathbf{M}_r\mathbf{Q})(\mathbf{U}_r-\mathbf{N}_r\mathbf{Q})^{-1}|~\mathbf{Q} \in \mathcal{RH}_\infty^{m \times p}\}\footnote{Equivalently, $\mathcal{C}_{\text{stab}}=\{(\mathbf{U}_l-\mathbf{QN}_l)^{-1}(\mathbf{V}_l-\mathbf{QM}_l)|~\mathbf{Q} \in \mathcal{RH}_\infty^{m \times p}\}$.}\,.
\end{equation}
%
%
Furthermore, it was proved in \cite[Chapter 4]{francis1987course} that the set of all closed-loop transfer functions from $w$ to $z$ achievable by $\mathbf{K} \in \mathcal{C}_{\text{stab}}$ is
\begin{align*}
f(\mathcal{C}_{\text{{stab}}})=\{\mathbf{T}_1-\mathbf{T}_2\mathbf{Q}\mathbf{T}_3\rvert ~\mathbf{Q} \in \mathcal{RH}_\infty^{m \times p} \}\,,
\end{align*}
where $f(\cdot)$ is defined in (\ref{eq:closed_loop_cost}) and $\mathbf{T}_1=\mathbf{P}_{11}+\mathbf{P}_{12}\mathbf{V}_r\mathbf{M}_l\mathbf{P}_{21}$, $\mathbf{T}_2=\mathbf{P}_{12}\mathbf{M}_r$ and $\mathbf{T}_3=\mathbf{M}_l\mathbf{P}_{21}$.
    To facilitate our problem formulation, we define
\begin{align}
&\mathbf{Y}_Q=(\mathbf{V}_r-\mathbf{M}_r\mathbf{Q})\mathbf{M}_l\,, \label{eq:Ydoublycoprime}\\
&\mathbf{X}_Q=(\mathbf{U}_r-\mathbf{N}_r\mathbf{Q})\mathbf{M}_l \label{eq:Xdoublycoprime}\,.
\end{align}
  It directly follows from \eqref{eq:Cstab} that
    \begin{equation}
  \label{eq:Cstab_YX}
\mathcal{C}_{\text{stab}}=\{\mathbf{Y}_Q\mathbf{X}_Q^{-1}\rvert\;\eqref{eq:Ydoublycoprime},\eqref{eq:Xdoublycoprime},~\mathbf{Q} \in \mathcal{RH}_\infty^{m \times p}\}\,.
\end{equation}
We notice that (\ref{eq:dp2}) implies $\mathbf{U}_r=\mathbf{M}_l^{-1}+\mathbf{GV}_r$ and (\ref{eq:Ydoublycoprime}) implies $\mathbf{V}_r\mathbf{M}_l=\mathbf{Y}_Q+\mathbf{M}_r\mathbf{Q}\mathbf{M}_l$. Hence, we have 
\begin{align}
\mathbf{X}_Q&=(\mathbf{M}_l^{-1}+\mathbf{GV}_r-\mathbf{N}_r\mathbf{Q})\mathbf{M}_l \nonumber\\
&=I_p+\mathbf{G}(\mathbf{Y}_Q+\mathbf{M}_r\mathbf{QM}_l)-\mathbf{N}_r\mathbf{QM}_l\nonumber\\
&=I_p+\mathbf{GY}_Q\,. \label{eq:XisIplusGY}
\end{align}


Now we can equivalently reformulate $\mathcal{P}_K$ into the following optimization problem.
\begin{alignat*}{3}
&\textbf{Problem }\mathcal{P}_{Q}\\
	 &\minimize_{\mathbf{Q} \in \mathcal{RH}_\infty^{m \times p}} ~~&& \|\mathbf{T}_1-\mathbf{T}_2\mathbf{QT}_3\|\\
	&\st~&&(\ref{eq:Ydoublycoprime}),~~ (\ref{eq:Xdoublycoprime}),~~ \mathbf{Y}_Q \mathbf{X}_Q^{-1} \in \text{Sparse}(S)\,.
\end{alignat*}	
Without the sparsity constraint $\text{Sparse}(S)$,  problem $\mathcal{P}_{Q}$ would be convex, as (\ref{eq:Ydoublycoprime}), (\ref{eq:Xdoublycoprime}) and the cost function are affine in $\mathbf{Q}$. The primary source of non-convexity is  the requirement that $\mathbf{Y}_Q\mathbf{X}_Q^{-1} \in \text{Sparse}(S)$. We conclude that the complexity of distributed control is ultimately linked to  the non-convex sparsity requirement on the Youla parameter. 
%


\section{Sparsity Invariance}
\label{se:sparsity}

One approach to remove the non-convex sparsity requirement on the Youla parameter is as follows:  replace 
$\mathbf{Y}_Q\mathbf{X}_Q^{-1} \in \text{Sparse}(S)$ with the convex constraint that $\mathbf{Y}_Q$ and $\mathbf{X}_Q$ comply with appropriate sparsity patterns, in a way such that $\mathbf{Y}_Q\mathbf{X}_Q^{-1}$ is guaranteed to lie in $\text{Sparse}(S)$. In other words, we restrict our attention to  distributed sparse controllers $\mathbf{K} \in \text{Sparse}(S)$ defined as the product of two structured matrix factors. We note that related ideas appeared for the specific case of row-column sparsities (e.g. \cite{matni2016regularization, dhingra2014admm}), but the case of arbitrary sparsities was not addressed.

Following the general idea above, in this paper we investigate a  notion of Sparsity Invariance (SI) for convex design of sparse controllers. As will be thoroughly discussed in Section~\ref{se:SparsityInvarianceQI}, SI leads to the largest known class of convex restrictions of $\mathcal{P}_K$ for general systems subject to sparsity constraints on the controller.


\begin{definition}[Sparsity Invariance (SI)]
    Given a binary matrix $S$, the pair of binary matrices $T, R$ satisfies a property of sparsity invariance (SI) with respect to $S$ if
\begin{align}
\label{eq:sparsity_invariance}
&\mathbf{Y} \in \text{Sparse}(T) \text{ and }\mathbf{X} \in \text{Sparse}(R)\nonumber\\
&\qquad \qquad \qquad \quad \Downarrow\\
& \qquad \quad~~\mathbf{YX}^{-1} \in \text{Sparse}(S). \nonumber
\end{align}
\end{definition}

Motivated by the SI property, consider the following convex problem:
	\begin{empheq}[box=\fbox]{alignat*=3}
	&\textbf{Problem }&&\mathcal{P}_{T,R}~~ \nonumber\\
	&\minimize_{\mathbf{Q}\in \mathcal{RH}_\infty^{m \times p}} &&~\|\mathbf{T}_1-\mathbf{T}_2\mathbf{QT}_3\| \nonumber\\
 &\st && ~(\ref{eq:Ydoublycoprime}),~  (\ref{eq:Xdoublycoprime})\,,\\
 &~&& ~\mathbf{Y}_Q \mathbf{\Gamma} \in \text{Sparse}(T), ~ \mathbf{X}_Q \mathbf{\Gamma} \in \text{Sparse}(R)\,,
	\end{empheq}
where $T \in \{0,1\}^{m \times p}, R \in \{0,1\}^{p \times p}$ and $\mathbf{\Gamma} \in \mathcal{R}_c^{p \times p}$, with $\mathbf{\Gamma}$ invertible, are parameters to be designed before performing the optimization.  For simplicity,  one could select $\mathbf{\Gamma}=I_p$, but we illustrate in Example~\ref{ex:1} of Section~\ref{se:SparsityInvarianceQI} that there are cases where a different choice of $\mathbf{\Gamma}$ might lead to improved and even globally-optimal performance for non-QI problems. 
For any choice of $T,~R$ and $\mathbf{\Gamma}$, the above program is convex. One fundamental question is when its feasible solutions  lead to stabilizing controllers $\mathbf{K}=(\mathbf{Y}_Q\mathbf{\Gamma})(\mathbf{X}_Q\mathbf{\Gamma})^{-1}=\mathbf{Y}_Q\mathbf{X}_Q^{-1}$ lying in the desired sparsity subspace $\text{Sparse}(S)$. The notion of SI (\ref{eq:sparsity_invariance}) defined above is a mathematical expression of this requirement. 

In the next subsection we establish necessary and sufficient conditions on the binary matrices $T$ and $R$ to satisfy the SI property (\ref{eq:sparsity_invariance}).

\begin{remark}
 \emph{Note that the notion of SI is an algebraic requirement for binary matrices $R$ and $T$, given a binary matrix $S$. This is independent of the parameterization of internally stabilizing controllers. In addition to the Youla parameterization, 
we recently observed that the SI idea \eqref{eq:sparsity_invariance} is equivalently applicable within the  system-level \cite{wang2019system} (SLP) and input-output \cite{furieri2019input} (IOP) parameterizations, in both continuous- and discrete-time. We refer to \cite[Remark 4]{zheng2019equivalence} for details. For brevity, in this paper we will develop our theoretical results within the Youla parameterization, and note that they can be straightforwardly applied to the SLP and the IOP.}
\end{remark}

\begin{remark}
\emph{We assume that $R\geq I_p$. Since $\mathbf{X}_Q=I_p+\mathbf{GY}_Q\in \text{Sparse}(R)$ and $\mathbf{G}$ is strictly proper, the assumption is without loss of generality for $\mathbf{\Gamma}=I_p$. For convenience, in the definition of problem $\mathcal{P}_{T,R}$ we do not indicate $\mathbf{\Gamma}$ explicitly as a parameter. This is because the SI property (\ref{eq:sparsity_invariance})  only depends on the binary matrices $T$ and $R$.}
\end{remark}


\subsection{Characterization of SI}
One immediate idea in designing the binary matrices $T$ and $R$ to guarantee $\mathbf{K}=(\mathbf{Y}_\mathbf{Q}\mathbf{\Gamma})(\mathbf{X}_\mathbf{Q} \mathbf{\Gamma})^{-1}=\mathbf{Y}_Q\mathbf{X}_Q^{-1} \in \text{Sparse}(S)$ is to simply select $T=S$ and $R=I_p$ similar to  \cite{conte2012distributed, geromel1994decentralized, zheng2017convex}. However, many other choices are available that lead to improved convex restrictions. 

The next Theorem provides a full characterization of  the SI property (\ref{eq:sparsity_invariance}) in terms of the binary matrices $T$ and $R$.    

\begin{theorem}
\label{th:sparsity_invariance}
Let $T \in \{0,1\}^{m \times p}$ and $R \in \{0,1\}^{p \times p}$ be such that $R\geq I_p$. The following two statements are equivalent:
\begin{enumerate}
\item $T\leq S$ and $TR^{p-1} \leq S$.
\item SI as per (\ref{eq:sparsity_invariance}) holds.
\end{enumerate}
\end{theorem}
The proof of Theorem~\ref{th:sparsity_invariance} is reported in Appendix~\ref{appsub:theorem}. The relevance of Theorem~\ref{th:sparsity_invariance} to characterizing a class of convex restrictions of $\mathcal{P}_K$ is stated in the following Corollary.

\begin{corollary}
\label{co:SI}
Let $T \in \{0,1\}^{m \times p}$ and $R \in \{0,1\}^{p \times p}$ be such that $R\geq I_p$, $T\leq S$ and $TR^{p-1} \leq S$. Then, problem $\mathcal{P}_{T,R^{p-1}}$ is a convex restriction of $\mathcal{P}_K$ for any invertible transfer matrix $\mathbf{\Gamma} \in \mathcal{R}_c^{p \times p}$.
\end{corollary}
 \begin{proof}
\emph{Problem $\mathcal{P}_{T,R^{p-1}}$ is obviously convex. We only need to show that any solution to $\mathcal{P}_{T,R^{p-1}}$ corresponds to a feasible solution of $\mathcal{P}_Q$.} 

\emph{First, given any invertible $\mathbf{\Gamma} \in \mathcal{R}_{c}^{p \times p}$ we have $$(\mathbf{Y}_Q\mathbf{\Gamma})(\mathbf{X}_Q\mathbf{\Gamma})^{-1}=\mathbf{Y}_Q\mathbf{X}_Q^{-1}.$$ 
Let $\mathbf{Y}=\mathbf{Y}_Q\mathbf{\Gamma}$ and $\mathbf{X}=\mathbf{X}_Q\mathbf{\Gamma}$ in (\ref{eq:sparsity_invariance}). Since (\ref{eq:sparsity_invariance}) holds by Theorem~\ref{th:sparsity_invariance}, by definition $\mathbf{YX}^{-1}=\mathbf{Y}_Q\mathbf{X}_Q^{-1} \in \text{Sparse}(S)$ and thus every solution of $\mathcal{P}_{T,R}$ is a solution of $\mathcal{P}_Q$.}

\emph{Second, since $\mathcal{P}_Q$ is equivalent to $\mathcal{P}_K$, we conclude that  $\mathcal{P}_{T,R}$ is a restriction of $\mathcal{P}_K$ for every invertible $\mathbf{\Gamma} \in \mathcal{R}_c^{p \times p}$.}

\emph{Finally, since $TR^{p-1}\leq S$ and $R\geq I_p$ we have that  $T(R^{p-1})^{p-1} \leq S$ by transitive closure of the graph having $R$ as its adjacency matrix. Hence, $\mathcal{P}_{T,R^{p-1}}$ is a convex restriction of $\mathcal{P}_K$ for every invertible $\mathbf{\Gamma} \in \mathcal{R}_c^{p \times p}$.}
 \end{proof}

In summary, the algebraic conditions
\begin{equation}
\label{eq:TR<S}
 T\leq S\text{ and }TR^{p-1}\leq S\,,
\end{equation}
are equivalent to SI and yield  a class of convex restrictions of $\mathcal{P}_K$. 
  Clearly, our condition (\ref{eq:TR<S}) includes the choice $T=S$ and $R$ is (block)-diagonal as per \cite{geromel1994decentralized, zheng2017convex, conte2012distributed}. We will further show in Section~\ref{se:SparsityInvarianceQI} that the convex restrictions developed in \cite{rotkowitz2012nearest} are a particular case of (\ref{eq:TR<S}). Therefore, our notion of SI naturally encompasses and extends previous convex restrictions of $\mathcal{P}_K$.

\begin{remark}
\emph{For each $T$ and $R$ as per (\ref{eq:TR<S}), it is always preferable to  solve the convex restriction $\mathcal{P}_{T,R^{p-1}}$ instead of $\mathcal{P}_{T,R}$. Indeed, notice that since $TR^{p-1} \leq S$ and $R\geq I_p$, then $T(R^{p-1})^{p-1} \leq S$. Equivalently, when $T$ and $R$ satisfy sparsity invariance (\ref{eq:TR<S}), so do $T$ and $R^{p-1}$, and both $\mathcal{P}_{T,R}$ and $\mathcal{P}_{T,R^{p-1}}$ are convex restrictions of $\mathcal{P}_K$. Since requiring  $\mathbf{X}_Q \in \text{Sparse}(R')$ for some $R' < R^{p-1}$ may be conservative in the case $\text{Sparse}(R') \subset \text{Sparse}(R^{p-1})$, we will focus on the convex restriction $\mathcal{P}_{T,R^{p-1}}$ to avoid this possibility.}
\end{remark}

After determining all the matrices $T$ and $R$ for sparsity invariance, a natural follow-up question arises: how can we choose $T$ and $R$ as per Theorem~\ref{th:sparsity_invariance} to obtain a convex restriction of $\mathcal{P}_K$ that is as tight as possible?

\subsection{ Optimized design of SI}
\label{se:optimal}

Here, we study how to choose the binary matrices 
$T$ and $R$ optimally for a fixed invertible $\mathbf{\Gamma}\in \mathcal{R}_c^{p \times p}$. 

  In order to determine the best performing choice for $T$ and $R$ satisfying (\ref{eq:TR<S}), one would need in general to solve $\mathcal{P}_{T,R^{p-1}}$ with the chosen $\mathbf{\Gamma}$ for each $T$ and $R$  such that (\ref{eq:TR<S}) holds, and then select the problem minimizing the objective $\|\mathbf{T}_1-\mathbf{T}_2\mathbf{Q}\mathbf{T}_3\|$.  Clearly, this approach is not tractable in general, as one needs to solve a large number of convex programs that is exponential in $m$ and $p$, that is, one convex program for each binary matrices $T$ and $R$ such that $TR^{p-1}\leq S$. Even if we simplify the search above by fixing any $T\leq S$ and looking for the best performing choice of $R$, we would still need to solve a large number of convex programs that is exponential in $p$, that is, one convex program for each binary matrix $R$ such that $TR^{p-1} \leq S$. To deal with the above challenges,  here we suggest a suboptimal, but  computationally efficient algorithm that  generates a locally optimized binary matrix $R$ tailored to any chosen $T\leq S$. 

Specifically, 
our proposed approach is to $1)$ select $T\leq S$ and then $2)$ compute that binary matrix  $R^\star_T$ which is the least sparse among those satisfying 
\begin{equation}
\label{eq:restriction}
T{R}^{p-1}\leq T\,.
\end{equation}
Clearly, both $1)$ and $2)$ above are simplifications of the general problem of finding the globally tightest convex restriction $\mathcal{P}_{T,R}$ of $\mathcal{P}_K$ for a fixed invertible $\mathbf{\Gamma} \in \mathcal{R}_c^{p \times p}$; indeed, we do not optimize over $T$ and we impose \eqref{eq:restriction}, a condition stronger than the SI requirement \eqref{eq:TR<S}. 
The gain is that $R^\star_T$ is unique and can be computed efficiently as per Algorithm~\ref{alg:procedure}, which has a polynomial complexity of $\mathcal{O}(mp^2)$. 

	\begin{algorithm}[H]
	\caption{Generation of $R^\star_T$}
	\label{alg:procedure}
  \begin{algorithmic}[1]
  \State{Initialize $R^\star_T=1_{p \times p}$}
  \For{ each $i =1,\ldots, m$, $k=1,\ldots, p$}
  	\If{$T_{ik}==0$}
  	\For{ each $j=1, \ldots, p$}
  	\If{$T_{ij}==1$}
  	\State{$(R^\star_T)_{jk}\leftarrow 0$}
  	\EndIf
  	\EndFor
  	\EndIf
  \EndFor
  \end{algorithmic}
  \end{algorithm}
  
The idea behind Algorithm~\ref{alg:procedure} is to only set an entry of $R_T^\star$ to 0 if the condition $TR_T^\star \leq T$ would be violated. We now formalize the main result 
about $R_T^\star$ .

\begin{theorem}
\label{th:optimal_binary}
    Consider a binary matrix $T \in \{0,1\}^{m \times p}$, and define $\mathcal{R}_T := \{R \in \{0,1\}^{p \times p} \mid R \geq I_p, (\ref{eq:restriction})\text{ holds}\}$. Then,
    \begin{enumerate}
      \item There exists a unique $R^\star_T\in \mathcal{R}_T$ such that
    $
        R^\star_T \geq R^{p-1}, \forall R \in \mathcal{R}_T.
    $
      \item Such  $R^\star_T$ can be computed via Algorithm~\ref{alg:procedure}.
    \end{enumerate}
\end{theorem}
\begin{proof}
\emph{Let $R^\star_T$ be the unique binary matrix generated by Algorithm~\ref{alg:procedure}. It is easy to check that $TR^\star_T \leq T$ by construction.  Since $R^\star_T\geq I_p$, it follows $(TR^\star_T)R^\star_T \leq TR^\star_T \leq T$ and $T(R^\star_T)^{p-1} \leq \cdots \leq TR^\star_T\leq T$. We conclude $R^\star_T \in \mathcal{R}_T$.}

\emph{Next, consider any binary matrix $R \in \mathcal{R}_T$. By definition, we have that $TR^{p-1} \leq T$ and so $(R^{p-1})_{jk}=0$ whenever $T_{ij}=1$ and $T_{ik}=0$. Then, $R^{p-1} \leq R^\star_T$ since $(R^\star_T)_{jk}$ is set to $0$ by Algorithm~\ref{alg:procedure} if and only if $T_{ik}=0$ and $T_{ij}=1$. Therefore, we have $R^{p-1} \leq R^\star_T$, $\forall R \in \mathcal{R}_T$.}
\end{proof}

The next corollary connects our result to characterizing tight convex restrictions of $\mathcal{P}_K$.
\begin{corollary}
\label{co:procedure}
Given a binary matrix $T\leq S$, compute $R^\star_T$  as per Algorithm~\ref{alg:procedure}. Then, for every fixed invertible $\mathbf{\Gamma} \in \mathcal{R}_c^{p \times p}$, $\mathcal{P}_{T,R^\star_T}$ is the tightest convex restriction of $\mathcal{P}_K$ among those in the form $\mathcal{P}_{T,R^{p-1}}$ with $R \in \mathcal{R}_T$.
\end{corollary}
\begin{proof}
\emph{Fix an invertible $\mathbf{\Gamma} \in \mathcal{R}_c^{p \times p}$ and consider the problems $\mathcal{P}_{T,R^{p-1}}$ and $\mathcal{P}_{T,R^\star_T}$, where $R \in \mathcal{R}_T$ and $R^\star_T$ is generated by Algorithm~\ref{alg:procedure}. By Theorem~\ref{th:optimal_binary}, we have $R^{p-1} \leq R^\star_T$, meaning that $\text{Sparse}(R^{p-1}) \subset \text{Sparse}(R^\star_T)$. }

\emph{The only difference between problem $\mathcal{P}_{T,R^{p-1}}$ and problem $\mathcal{P}_{T,R^\star_T}$ is: $\mathcal{P}_{T,R^{p-1}}$ requires  $\mathbf{X}_Q\mathbf{\Gamma} \in \text{Sparse}(R^{p-1})$ while $\mathcal{P}_{T,R^\star_T}$ requires $\mathbf{X}_Q\mathbf{\Gamma} \in \text{Sparse}(R^\star_T)$. Therefore, we conclude that $\mathcal{P}_{T,R^\star_T}$ admits the largest feasible region among all $\mathcal{P}_{T,R^{p-1}}$ with $R \in \mathcal{R}_T$. This completes our proof.}
\end{proof}


Our suggested procedure can find a tight convex restriction for $\mathcal{P}_K$ by using the computationally efficient Algorithm~1, which makes the approach practical for  practitioners. However,  optimally choosing $\mathbf{\Gamma}$ and $T$ is also a non-trivial task which we leave for future work. We remark that in the lack of any further insight, one can always choose $T=S$ and $\mathbf{\Gamma}=I_p$ and still obtain sparse controllers with tight sub-optimality gaps, as will be shown experimentally in Section~\ref{se:numerical}. Furthermore, as  shown in Section~\ref{se:SparsityInvarianceQI}, the trivial choice $T=S$ and $\mathbf{\Gamma}=I_p$ combined with Algorithm~\ref{alg:procedure} for choosing $R$ is sufficient to recover and extend the optimality results of \cite{rotkowitz2006characterization}, \cite{rotkowitz2012nearest}  which are based on the Quadratic Invariance (QI) notion.  We conclude this section by providing an example to illustrate the SI approach.

\begin{example}
\emph{Motivated by the numerical example in \cite{rotkowitz2006characterization}, let us consider the unstable plant
\begin{equation*}
\mathbf{G}=\begin{bmatrix}
u(\sigma)&0&0&0&0\\
u(\sigma)&v(\sigma)&0&0&0\\
u(\sigma)&v(\sigma)&u(\sigma)&0&0\\
u(\sigma)&v(\sigma)&u(\sigma)&u(\sigma)&0\\
u(\sigma)&v(\sigma)&u(\sigma)&u(\sigma)&v(\sigma)
\end{bmatrix}\,,
\end{equation*}
with $u(\sigma)=u(s)=\frac{1}{s+1}$, $v(\sigma)=v(s)=\frac{1}{s-1}$ in continuous-time or $u(\sigma)=u(z)=\frac{0.1}{z-0.5}$, $v(\sigma)=v(z)=\frac{1}{z-2}$ in discrete-time}, and define
\begin{equation*}
\mathbf{P}_{11}=\begin{bmatrix}
\mathbf{G}&0_{5\times 5}\\0_{5\times 5}&0_{5\times 5}
\end{bmatrix}\,,\quad \mathbf{P}_{12}=\begin{bmatrix}
\mathbf{G}\\I_5
\end{bmatrix}\,,\quad \mathbf{P}_{21}=\begin{bmatrix}
\mathbf{G}&I_5
\end{bmatrix}\,.
\end{equation*}
Our goal is to design a stabilizing controller $\mathbf{K}$ which minimizes  $\|f(\mathbf{K})\|_{\mathcal{H}_2}$ and satisfies the sparsity pattern below:
\begin{equation*}
S=\begin{bmatrix}
1&0&0&0&0\\1&1&0&0&0\\0&1&1&0&0\\0&1&1&1&0\\0&1&1&1&1
\end{bmatrix}\,.
\end{equation*}
This information structure is depicted in Figure~\ref{fig:example}.
\label{ex:1}
	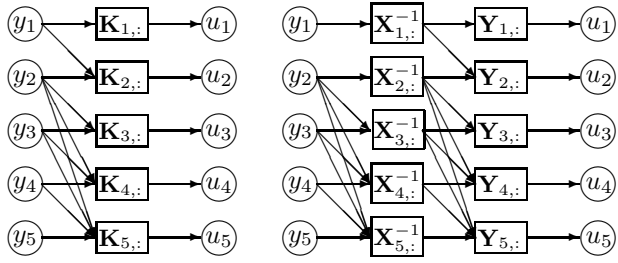
\begin{figure}[t]
	      \centering
	      \captionsetup{justification=raggedright} 
\begin{center}\input{figs/fig_infostructure}\end{center}
    \caption{\footnotesize  In the figure, we denote  as $\mathbf{K}_{i,:}$, $\mathbf{Y}_{j,:}$, $\mathbf{X}^{-1}_{k,:}$ the $i$th, $j$th and $k$th row of $\mathbf{K}$, $\mathbf{Y}_Q$ and $\mathbf{X}^{-1}_Q$ respectively. For every non-zero entry of $\mathbf{K}_{i,:}$, $\mathbf{Y}_{j,:}$ or $\mathbf{X}^{-1}_{k,:}$ the corresponding signal enters the block with an arrow, thus representing the information flow from measured outputs to control signals. The scheme on the left represents the desired sparsity pattern $S$ for controller $\mathbf{K}$. The scheme on the right represents the sparsity pattern of controllers that are feasible for $\mathcal{P}_{S,R^\star_S}$, i.e. those in the form $\mathbf{Y}_Q(\mathbf{X}_Q)^{-1}$ with $\mathbf{Y}_Q \in \text{Sparse}(S)$ and $\mathbf{X}_Q \in \text{Sparse}(R^\star_S)$.}
    \label{fig:example}
	\end{figure}

  \emph{Here, we apply the proposed SI approach and Algorithm~\ref{alg:procedure} for sparsity design in order to obtain a convex restriction of $\mathcal{P}_K$. For this instance, we choose to fix $T=S$ and $\mathbf{\Gamma}=I_p$. According to Theorem~\ref{th:optimal_binary} and Corollary~\ref{co:procedure}, the tightest convex restriction of $\mathcal{P}_K$ such that $TR^{p-1}=SR^{p-1}\leq S$ is $\mathcal{P}_{S,R^\star_S}$, where $R^\star_S$
\begin{equation*}
R^\star_S=\begin{bmatrix}
1&0&0&0&0\\
0&1&0&0&0\\
0&1&1&0&0\\
0&1&1&1&0\\
0&1&1&1&1
\end{bmatrix}\,,
\end{equation*}
 is generated via Algorithm~\ref{alg:procedure}. Given a doubly coprime factorization of $\mathbf{G}$, any solution of $\mathcal{P}_{S,R^\star_S}$  is in the form $\mathbf{K}=\mathbf{Y}_Q(\mathbf{X}_Q)^{-1} \in \mathcal{C}_{\text{stab}} \cap \text{Sparse}(S)$, where $\mathbf{Y}_Q \in \text{Sparse}(S)$, $\mathbf{X}_Q \in \text{Sparse}(R^\star_S)$ and $(\mathbf{X}_Q)^{-1} \in \text{Sparse}(R^\star_S)$.} 
\begin{remark}[Performance improvement]
\emph{The classical immediate idea would be to require that $\mathbf{X}_Q$ is diagonal as per  \cite{conte2012distributed, geromel1994decentralized, zheng2017convex}; instead, SI allows the off-diagonal entries of $\mathbf{X}_Q=I_p+\mathbf{GY}_Q$ to be non-zero through the optimized choice of $R^\star_S$, thus removing unnecessary constraints on the entries of $(\mathbf{GY}_Q)$.}
 \emph{This additional freedom can be seen graphically on the right side of Figure~\ref{fig:example}; the information flow from outputs to control inputs remains the same as the one encoded by $S$, but we allow  for as many arrows as possible in the first stage from outputs to the rows of $\mathbf{X}^{-1}$, thus maximizing the degrees of freedom in the optimization. In Section~\ref{se:numerical} we will numerically solve $\mathcal{P}_{S,R^\star_S}$ for this example and show that performance improvement over the method of \cite{rotkowitz2012nearest} is obtained.}
 \end{remark}
\end{example}


\section{Beyond Quadratic Invariance}
\label{se:SparsityInvarianceQI}
We start by recalling the well-known notion of Quadratic Invariance (QI) \cite{rotkowitz2006characterization} in Subsection~\ref{sub:QI}, and its application to the design of globally optimal \cite{rotkowitz2006characterization} and sub-optimal \cite{rotkowitz2012nearest} distributed dynamic output-feedback controllers in Subsection~\ref{sub:conv_restrictions_QI}. In Subsections~\ref{sub:extending_previous}, \ref{sub:strictly_beyond} we show that the suggested SI notion  strictly goes beyond that of QI for sparsity constraints:  1) the controllers obtained using the SI notion perform at least as well as those obtained by \cite{rotkowitz2006characterization} and \cite{rotkowitz2012nearest}; 2) we show through examples that using the SI notion we can recover globally optimal controllers even when QI does not hold, and that strict performance improvements over \cite{rotkowitz2012nearest} can be obtained in general. Last, in Subsection~\ref{sub:static}, we discuss the applicability of SI to computing distributed static controllers, whereas the QI notion is not applicable.

\subsection{Quadratic Invariance}
\label{sub:QI}
The celebrated work of \cite{rotkowitz2006characterization} characterized conditions on $\mathbf{G}$ and $\text{Sparse}(S)$ under which $\mathcal{P}_K$ admits an exact convex reformulation in the Youla parameter $\mathbf{Q}$, denoted as \emph{quadratic invariance} (QI). 
\begin{definition}[Quadratic invariance~\cite{rotkowitz2006characterization}]
A subspace $\mathcal{K} \subseteq \mathcal{R}_c^{m \times p}$ is QI with respect to $\mathbf{G}$ if
\end{definition}
\begin{equation*}
\mathbf{KGK} \in \mathcal{K}\,, \quad \forall \mathbf{K} \in \mathcal{K}\,.
\end{equation*}

For the purpose of this paper, we will limit our focus to QI sparsity subspaces in the form $\text{Sparse}(S)$. It is shown that given a controller $\mathbf{K}_{\text{nom}} \in \text{Sparse}(S)$ that stabilizes $\mathbf{G}$ and is itself stable,  there exists a parametrization such that $\mathbf{K} \in \text{Sparse}(S) \Leftrightarrow \mathbf{Q} \in \text{Sparse}(S)$~\cite{rotkowitz2006characterization}. Accordingly, a convex optimization problem equivalent to $\mathcal{P}_K$ is obtained. The requirement of a stable and stabilizing controller $\mathbf{K}_{\text{nom}}$  was removed in \cite{sabuau2014youla}. One main result from \cite{sabuau2014youla} is as follows:
\begin{theorem}[Theorem IV.2 of \cite{sabuau2014youla}]
\label{th:QI_Youla}
Consider any doubly-coprime factorization of $\mathbf{G}$ and let $\text{Sparse}(S)$ be QI with respect to $\mathbf{G}$. Then, the following two statements hold:
\begin{enumerate}
\item If $\mathbf{Q} \in \mathcal{RH}_\infty^{m \times p}$ is such that $\mathbf{Y}_Q\in \text{Sparse}(S)$,
then $\mathbf{K}=\mathbf{Y}_Q\mathbf{X}_Q^{-1}$ is a stabilizing controller in $\text{Sparse}(S)$.
\item For any $\mathbf{K} \in \mathcal{C}_{\text{stab}} \cap \text{Sparse}(S)$ there exists $\mathbf{Q} \in \mathcal{RH}_\infty^{m \times p}$ for which $\mathbf{Y}_Q \in \text{Sparse}(S)$ and $\mathbf{K}=\mathbf{Y}_Q\mathbf{X}_Q^{-1}$.
\end{enumerate}
\end{theorem}
 According to Theorem~\ref{th:QI_Youla}, if $\text{Sparse}(S)$ is QI with respect to $\mathbf{G}$, then $\mathcal{P}_K$ can be equivalently reformulated as
    \begin{alignat}{3}
  &\minimize_{\mathbf{Q} \in \mathcal{RH}_\infty^{m \times p}} \quad && \|\mathbf{T}_1-\mathbf{T}_2\mathbf{QT}_3\| \label{eq:problem_QI}\\
	&\st \quad &&(\ref{eq:Ydoublycoprime}),\; (\ref{eq:Xdoublycoprime}),\; \mathbf{Y}_Q  \in \text{Sparse}(S)\,.\nonumber
    \end{alignat}
The optimal solution $\mathbf{Q}^\star$ of (\ref{eq:problem_QI}) can be used to recover the \emph{globally} optimal solution $\mathbf{K}^\star$ of  $\mathcal{P}_K$ via $\mathbf{K}^\star=\mathbf{Y}_{Q^\star}\mathbf{X}_{Q^\star}^{-1}$.
\subsection{Convex restrictions for non-QI sparsity patterns}
\label{sub:conv_restrictions_QI}
When $\text{Sparse}(S)$ is not QI with respect to $\mathbf{G}$, the authors of \cite{rotkowitz2012nearest} proposed finding a binary matrix $T_{\text{QI}} < S$ such that $\text{Sparse}(T_{\text{QI}})$ is QI with respect to $\mathbf{G}$. Then, the constraint $\mathbf{Y}_Q\mathbf{X}_Q^{-1} \in \text{Sparse}(S)$ of problem $\mathcal{P}_{Q}$ can  be replaced by $\mathbf{Y}_Q \in \text{Sparse}(T_{\text{QI}})$, and any feasible $\mathbf{Q}$ for this convex program will correspond to a feasible controller 
\begin{equation} \label{eq:nonQIrelation}
\begin{aligned}
\mathbf{K}=\mathbf{Y}_Q\mathbf{X}_Q^{-1} &\in \mathcal{C}_{\text{stab}} \cap \text{Sparse}(T_{\text{QI}})\\
& \subseteq \mathcal{C}_{\text{stab}} \cap\text{Sparse}(S)\,.
\end{aligned}
\end{equation}
This inclusion~\eqref{eq:nonQIrelation} directly follows from Theorem~\ref{th:QI_Youla} and the fact that $\text{Sparse}(T_{\text{QI}}) \subset \text{Sparse}(S)$.

 A challenge of this approach is to compute $T_{\text{QI}}$ such that $\text{Sparse}(T_{\text{QI}})$ is QI and as close as possible to $S$ in order to reduce conservatism, in the sense that $\|S\|_0-\|T_{\text{QI}}\|_0$ is minimized. In general, there might be multiple choices of $T_{\text{QI}}$ with the same cardinality. Furthermore,  the QI condition $T_{\text{QI}}\Delta T_{\text{QI}} \leq T_{\text{QI}}$ of \cite[Theorem~26]{rotkowitz2006characterization}, where $\Delta=\text{Struct}(\mathbf{G})$, is nonlinear in $T_{\text{QI}}$. For these reasons, a procedure to compute a closest QI subset of $S$ in polynomial time was not provided in \cite{rotkowitz2012nearest}.  Instead, we have shown that the polynomial time Algorithm~\ref{alg:procedure} can be combined with the SI notion to find a convex restriction for any given $T\leq S$. In the next subsections, we show that the recovered controllers perform at least as well as those based on the notion of QI by choosing $T \leq S$ appropriately, and can be strictly more performing in general even with the trivial choice $T=S$.

\subsection{Connections of SI with QI}
\label{sub:extending_previous}
Here, we show that it is not necessary to check the QI property in order to obtain a globally optimal solution. Note that checking the property of QI before solving $\mathcal{P}_{K}$ was proposed in \cite{rotkowitz2006characterization} and required in many subsequent works.  Indeed, the approach in \cite{rotkowitz2006characterization} is guaranteed to yield feasible solutions for $\mathcal{P}_K$ only if QI holds. Instead, our technique can be directly applied given $S$ without first checking QI. This result is summarized in the following theorem and corollary.

\begin{theorem}
\label{th:procedure+QI}
Let $\Delta=\text{Struct}(\mathbf{G})$ and let $R^\star_S$ be the binary matrix generated by Algorithm~\ref{alg:procedure} with $T=S$. The following statements are equivalent.
\begin{enumerate}[i)]
\item $\text{Sparse}(S)$ is QI with respect to $\mathbf{G}$.
\item $R^\star_S \geq I_p+\Delta S$, where $R^\star_S$ is generated by Algorithm~1 with $T=S$.
\end{enumerate}
\end{theorem}
\begin{proof}
\emph{i) $\Rightarrow$ ii): Suppose that $\text{Sparse}(S)$ is QI with respect to $\mathbf{G}$.  We have that $S\Delta S \leq S$ by \cite[Theorem 26]{rotkowitz2006characterization}, implying that $S(I_p+\Delta S) \leq S$ and ultimately  $$S(I_p+\Delta S)^{p-1} \leq S.$$  
In addition, we have that $R^\star_S \geq I_p$ and  $SR^\star_S\leq S$ by construction. It follows that $S(R^\star_S)^{p-1}\leq \ldots \leq SR^\star_S\leq S$. Also, according to Theorem~\ref{th:optimal_binary}, we have $R^\star_S \geq R$, $\forall R \geq I_p$ such that $SR^{p-1}\leq S$. By posing $R=I_p+\Delta S$, we have shown above that $SR^{p-1}\leq S$. Hence, $ R^\star_S\geq R=I_p+\Delta S$.}


\emph{ii) $\Rightarrow$ i):  Suppose that $R^\star_S \geq I_p+\Delta S$, which implies $(R^\star_S)^{p-1} \geq (I_p+\Delta S)^{p-1}$. By definition of $R^\star_S$, we have observed that  $S (R^\star_S)^{p-1}\leq S$.  It follows that
\begin{equation} 
\label{eq:QIproof_s2}
S(I_p+\Delta S)^{p-1} \leq S(R^\star_S)^{p-1} \leq S\,.
\end{equation}
Combining~\eqref{eq:QIproof_s2} with the fact that $(I_p+\Delta S) \geq I_p$, we have
\begin{equation*}
S(I_p+\Delta S) \leq S(I_p+\Delta S)^{p-1} \leq S\,.
\end{equation*}
This implies $S\Delta S \leq S$ which is equivalent to QI by \cite[Theorem 26]{rotkowitz2006characterization}.}
\end{proof}
\begin{corollary}\label{co:SI_QI}
 The following statements are equivalent.
\begin{enumerate}[i)]
\item $\text{Sparse}(S)$ is QI with respect to $\mathbf{G}$.
\item $\mathcal{P}_K$ is equivalent to $\mathcal{P}_{S,R^\star_S}$ with $\mathbf{\Gamma}=I_p$, where $R^\star_S$ is the binary matrix generated by Algorithm~\ref{alg:procedure} with $T=S$.
\end{enumerate}
\end{corollary}
\begin{proof}
\emph{It is well-known \cite{sabuau2014youla,QIconvexity} that (\ref{eq:problem_QI}) is equivalent to $\mathcal{P}_K$ if and only if QI holds. It remains to show that $\mathcal{P}_{S,R^\star_S}$ is equivalent to (\ref{eq:problem_QI}) if and only if QI holds. }

\emph{We first show that  $\mathbf{X}_Q$  lies in $\text{Sparse}(I_p+\Delta S)$ for every $\mathbf{Q} \in \mathcal{RH}_\infty^{m \times p}$ such that $\mathbf{Y}_Q \in \text{Sparse}(S)$. Indeed, by (\ref{eq:XisIplusGY}) we have $\mathbf{X}_Q=I_p+\mathbf{GY}_Q$ for every $\mathbf{Q} \in \mathcal{RH}_\infty^{m \times p}$ and thus $\mathbf{X}_Q \in \text{Sparse}(I_p+\Delta S)$. We have shown in Theorem~\ref{th:procedure+QI}  that QI is equivalent to $R^\star_S \geq I_p+\Delta S$, where $R^\star_S$ is generated by Algorithm~\ref{alg:procedure}. It follows that the constraint $\mathbf{Y}_Q \mathbf{\Gamma}=\mathbf{Y}_Q \in \text{Sparse}(S)$ makes the constraint $\mathbf{X}_Q\mathbf{\Gamma}=\mathbf{X}_Q \in \text{Sparse}(R^\star_S)$ redundant and thus $\mathcal{P}_{S,R^\star_S}$ with $\mathbf{\Gamma}=I_p$ is equivalent to (\ref{eq:problem_QI}). This concludes the proof.}
\end{proof}

Essentially, Theorem~\ref{th:procedure+QI} shows that QI is equivalent to $R^\star_S \geq I_p+\Delta S$. Since $\mathbf{X}_Q \in \text{Sparse}(I_p+\Delta S)$ by (\ref{eq:XisIplusGY}) when $\mathbf{Y}_Q \in \text{Sparse}(S)$, the constraint $\mathbf{X}_Q \in \text{Sparse}(R^\star_S)$ becomes redundant if and only if QI holds and the convex program we obtain with SI, namely $\mathcal{P}_{S,R^\star_S}$ with $\mathbf{\Gamma}=I_p$, is equivalent to $\mathcal{P}_K$ due to the results of \cite{rotkowitz2006characterization}. 

Theorems~\ref{th:sparsity_invariance},~\ref{th:optimal_binary} and~\ref{th:procedure+QI}, and Corollaries~\ref{co:SI}--\ref{co:SI_QI} can be summarized as follows. 
\begin{enumerate}
    \item Given any distributed sparsity-constrained control problem $\mathcal{P}_K$, one can always cast and solve its convex restriction $\mathcal{P}_{S,R^\star_S}$, where $R^\star_S$ is generated by Algorithm~\ref{alg:procedure}. 
\item If $\mathcal{P}_{S,R^\star_S}$ is feasible, its optimal solution is also feasible for $\mathcal{P}_K$, and is certified to be globally optimal if $\text{Sparse}(S)$ is QI with respect to $\mathbf{G}$. 
\end{enumerate}
We remark that verifying QI is optional and can be done a-posteriori to check global optimality of the solution, but QI is not part of the controller design procedure in the SI approach. Hence, Theorem~\ref{th:procedure+QI} expands the applicability of convex programming to compute distributed controllers for arbitrary systems and sparsity patterns, while maintaining previous global optimality results.

\begin{example}
\label{ex:QI_theorem}
\emph{Consider the unstable system and the sparsity pattern $S$ of Example~\ref{ex:1}. We can verify that $S\Delta S  \not \leq S$, where $\Delta=\text{Sparse}(\mathbf{G})$, and hence $\text{Sparse}(S)$ is not QI with respect to $\mathbf{G}$. Instead, let us consider  the new  sparsity pattern
\begin{equation}
\label{eq:S2}
S_2=\begin{bmatrix}
0&0&0&0&0\\0&1&0&0&0\\0&1&1&0&0\\0&1&1&1&0\\0&1&1&1&1
\end{bmatrix}\,.
\end{equation}
We can verify that $S_2\Delta S_2 \leq S_2$. Hence, $\text{Sparse}(S_2)$ is QI with respect to $\mathbf{G}$.  By applying Algorithm~\ref{alg:procedure} we obtain 
\begin{align*}
R^\star_{S_2}\hspace{-0.1cm}&=\hspace{-0.1cm}\begin{bmatrix}
1&1&1&1&1\\
0&1&0&0&0\\
0&1&1&0&0\\
0&1&1&1&0\\
0&1&1&1&1
\end{bmatrix}\hspace{-0.1cm}\,, ~I_p+\Delta S_2\hspace{-0.1cm}=\hspace{-0.1cm}\begin{bmatrix}
1&0&0&0&0\\
0&1&0&0&0\\
0&1&1&0&0\\
0&1&1&1&0\\
0&1&1&1&1
\end{bmatrix}\,,\\
R^\star_S\hspace{-0.1cm}&=\hspace{-0.1cm}\begin{bmatrix}
1&0&0&0&0\\
0&1&0&0&0\\
0&1&1&0&0\\
0&1&1&1&0\\
0&1&1&1&1
\end{bmatrix}\hspace{-0.1cm}\,,~~I_p+\Delta S\hspace{-0.1cm}=\hspace{-0.1cm}\begin{bmatrix}
1&0&0&0&0\\
\textcolor{red}{1}&1&0&0&0\\
\textcolor{red}{1}&1&1&0&0\\
\textcolor{red}{1}&1&1&1&0\\
\textcolor{red}{1}&1&1&1&1
\end{bmatrix}\,.
\end{align*}
In accordance with Theorem~\ref{th:procedure+QI} we have that $R^\star_{S_2} \geq I_p+\Delta S_2$, but $R^\star_S \not \geq I_p+\Delta S$ (see the entries highlighted in red). By Corollary~\ref{co:SI_QI}, we conclude that the convex program $\mathcal{P}_{S_2,R^\star_{S_2}}$ with $\mathbf{\Gamma}=I_p$ is equivalent to $\mathcal{P}_K$  with the sparsity constraint $\mathbf{K} \in \text{Sparse}(S_2)$, while $\mathcal{P}_{S,R^\star_S}$ is a convex restriction of $\mathcal{P}_K$ for every invertible $\mathbf{\Gamma} \in \mathcal{R}_c^{p \times p}$.}
\end{example}

Next, we show that SI generalizes the class of restrictions of \cite{rotkowitz2012nearest}, based on finding QI subsets of $\text{Sparse}(S)$ which are nearest to $\text{Sparse}(S)$. The result is a straightforward corollary of Theorem~\ref{th:procedure+QI}.
  \begin{corollary}
  \label{co:sparsity_subsets}
Let $\text{Sparse}(T_{\text{QI}})\subseteq \text{Sparse}(S)$ be QI with respect to $\mathbf{G}$ and let  $\|S\|_0-\|T_{\text{QI}}\|_0$ be minimal as proposed in \cite{rotkowitz2012nearest}.   Then, there exists $T\leq S$ such that $J^\star \leq J_{\text{QI}}$, where $J^\star$ is the minimum cost of $\mathcal{P}_{T,R^\star_T}$ with $\mathbf{\Gamma}=I_p$, and $J_{\text{QI}}$ is the minimum cost of problem (\ref{eq:problem_QI}) with the constraint $\mathbf{Y}_Q \in \text{Sparse}(S)$ replaced by $\mathbf{Y}_Q \in \text{Sparse}(T_{\text{QI}})$.
  \end{corollary}
  \begin{proof}
\emph{Let $T=T_{\text{QI}}$. Since $\text{Sparse}(T_{\text{QI}})$ is QI with respect to $\mathbf{G}$, we have $R^\star_T\geq I_p+\Delta T$ by Theorem~\ref{th:procedure+QI}. Hence, for every $\mathbf{Y}_Q\mathbf{\Gamma}=\mathbf{Y}_Q \in \text{Sparse}(T)$, the matrix $\mathbf{X}_Q=I_p+\mathbf{GY}_Q$ belongs to $\text{Sparse}(I_p+\Delta T)$ for every $\mathbf{Q} \in \mathcal{R}_\infty^{m\times p}$ and the constraint $\mathbf{X}_Q \mathbf{\Gamma}=\mathbf{X}_Q \in \text{Sparse}(R^\star_T)$ is redundant. It follows that  the choice $T=T_{\text{QI}}$ achieves $J^\star=J_{\text{QI}}$. Therefore, there exists a choice of $T$ such that the optimal solution of $\mathcal{P}_{T,R_T^\star}$ with $\mathbf{\Gamma}=I_p$  performs at least as well as that of the problem obtained by considering a nearest QI subset as suggested in \cite{rotkowitz2012nearest}. This completes our proof.}
  \end{proof}
  
  Corollary~\ref{co:sparsity_subsets} proves that the class of convex restrictions considered in \cite{rotkowitz2012nearest} is a special case in the framework of SI, obtained by choosing $T=T_{\text{QI}}$ and computing $R^\star_{T_{\text{QI}}}$ with our Algorithm~\ref{alg:procedure}. Furthermore, it is possible to choose $T\leq S$ to obtain strictly more performing convex restrictions, as we will show numerically in Section~\ref{se:numerical}.

  \subsection{Strictly Beyond QI}
\label{sub:strictly_beyond}
  So far, we have shown that the SI approach naturally recovers the previous QI results of \cite{rotkowitz2006characterization} and \cite{rotkowitz2012nearest} as specific cases by using Algorithm~\ref{alg:procedure}. Here and in Section~\ref{se:numerical}, we show through examples the stronger results that 
\begin{enumerate}
\item SI can recover globally optimal solutions when QI does \emph{not} hold,
\item   strictly better performance than the approach of \cite{rotkowitz2012nearest} can be obtained.
\end{enumerate}  
For point 2), we refer to the numerical results in Section~\ref{se:numerical}. For point 1), we consider an example taken from \cite{wang2019system}.
\begin{example}
\label{ex:Globally_Optimal}
\emph{Consider the optimal control problem:
\begin{alignat*}{3}
&\minimize_{\mathbf{K}(z)} &&\lim_{L \rightarrow \infty} \frac{1}{L} \sum_{t=0}^L \mathbb{E}||x(t)||_2^2\\
	&\text{subject to}&&~x(t+1)=Ax(t)+u(t)+w(t)\,,\\
	&~&&~u(z)=\mathbf{K}(z)x(z)\,, \quad \mathbf{K}(z) \in \text{Sparse}(A^{\text{bin}})\,,
\end{alignat*}
 where $z \in e^{j \mathbb{R}}$, $A \in \mathbb{R}^{n \times n}$, $A^{\text{bin}}=\text{Struct}(A)$ and  $w(t)$ denotes i.i.d. disturbances distributed according to a normal  distribution $\mathcal{N}(0_{n\times 1},I_n)$. The discrete-time transfer function of this system is $\mathbf{G}(z)=(zI_p-A)^{-1}$. This problem without the sparsity constraint on $\mathbf{K}$ is known as the LQR  problem. By adding the sparsity constraint, it is an instance of $\mathcal{P}_K$ in discrete-time. Notice that QI does not hold whenever the graph defined by $A$ is strongly connected because $\Delta=\text{Struct}(\mathbf{G}(z))=\text{Struct}\left((zI_n-A)^{-1}\right)$ is equal to $1_{n \times n}$ in general, and so $A^{\text{bin}}\Delta A^{\text{bin}} \not \leq A^{\text{bin}}$ thus violating QI.}
 
   \emph{The reason to consider a discrete-time instance of $\mathcal{P}_K$ is that one can solve analytically the corresponding problem where sparsity constraints are removed by computing a simple Riccati equation \cite{lancaster1995algebraic}. It so happens that the optimal solution for this problem is $\mathbf{K}(z)=-A$, which is also feasible and hence globally optimal for $\mathcal{P}_K$.  Now, consider problem $\mathcal{P}_{T,R}$ with $\mathbf{\Gamma}(z)=\mathbf{G}(z)$, $T=A^{\text{bin}}$ and $R=R^\star_{A^{\text{bin}}}$. We can verify that a feasible solution for $\mathcal{P}_{T,R}$ is $\mathbf{Y}_Q(z)=-\frac{A}{z}(zI_n-A)$, because $$\mathbf{Y}_Q\mathbf{\Gamma}=\mathbf{Y}_Q(zI_n-A)^{-1}=-\frac{A}{z} \in \text{Sparse}(A^{\text{bin}})\,.$$ This  implies $\mathbf{X}_Q(z)=I_n-\frac{A}{z}$ by (\ref{eq:XisIplusGY}). Hence, $\mathbf{X}_Q(z) \mathbf{\Gamma}(z)=\mathbf{X}_\mathbf{Q}(z)(zI_n-A)^{-1}=\frac{I_n}{z}$. Since $R^\star_{A^{\text{bin}}}\geq I_n$ by design (see Algorithm~\ref{alg:procedure}), we have $\mathbf{X}_Q(z) \mathbf{\Gamma}(z) \in \text{Sparse}(R^\star_{A^{\text{bin}}})$ as desired. It is immediate to verify that the resulting controller is $\mathbf{K}(z)=\mathbf{Y}_\mathbf{Q}(z)\mathbf{X}_Q(z)^{-1}=-A$. We conclude that, despite a lack of QI, a convex approximation which contains the global optimum of $\mathcal{P}_K$ is found by using the proposed SI approach.}
\begin{remark}
\emph{The global optimality result for this example was also obtained using the SLP in~\cite{wang2019system}. The sparsities for the system level parameters in \cite{wang2019system} were chosen empirically, while we provide an explicit methodology based on the SI condition (\ref{eq:TR<S}) and Algorithm~\ref{alg:procedure}. Furthermore, we wish to clarify that obtaining global optimality certificates for $\mathcal{P}_K$ for systems with non-QI constraints is still an open problem, which is not addressed neither by the system level approach \cite{wang2019system} nor by our SI approach. Both our approach and that of \cite{wang2019system} can certify optimality of the solution because the optimal solution of this simple instance is already known analytically.}
\end{remark}

\end{example}

\subsection{SI for static controller design}
\label{sub:static}
We conclude this section by highlighting another advantage of the SI notion over the QI notion; the SI notion can be used to compute sparse static control policies in a convex way, that is policies in the form $u(t)=Ky(t)$ where $K$ is a \emph{real} matrix in $\text{Sparse}(S)$. This topic has been thoroughly studied in our earlier work \cite{furieri2019separable}, where we derived a notion of SI limited to the static controller case. Here, we highlight that in contrast to the QI notion, SI is useful both for static and dynamic sparse controller design. 

The main observation is that the Youla parametrization cannot achieve a convexification of the static controller design problem in general, because enforcing $\mathbf{K}=(\mathbf{V}_r-\mathbf{M}_r\mathbf{Q})(\mathbf{U}_r-\mathbf{N}_r\mathbf{Q})^{-1}$ to be a $\emph{real}$ matrix is a non-convex requirement on the \emph{transfer} matrix $\mathbf{Q}$. Consequently, a different parametrization should be used and the QI property, tightly linked to using a Youla-like parametrization, will not be relevant anymore. The most well-known techniques to convexify the $\mathcal{H}_2$ and $\mathcal{H}_\infty$ norm-optimal state-feedback static controller design problems are based on computing appropriate quadratic Lyapunov functions through Linear Matrix Inequalities (LMI); see \cite{boyd1994linear,scherer2000linear} for a comprehensive review. The more general case of static output-feedback is known to be NP-hard \cite{Survey} and an exact convex formulation does not exist.

As we illustrated in \cite{furieri2019separable}, when the distributed static control problem is formulated through LMIs, the controller is recovered as $K=YX^{-1}$, where $Y$ and $X$ are real decision variables, $X$ is symmetric positive semidefinite and $V(x)=x^\mathsf{T}X^{-1}x$ is a quadratic Lyapunov function for the closed-loop system.  If the controller must lie in a sparsity subspace $\text{Sparse}(S)$, the only source of non-convexity stems from requiring that $YX^{-1}\in \text{Sparse}(S)$. This expression for the static controller in terms of the decision variables matches that of $\mathbf{K}=\mathbf{Y}_Q\mathbf{X}_Q^{-1}$, which is valid for dynamic controllers in terms of the Youla parameter. According to Theorem~\ref{th:sparsity_invariance} and Corollary~\ref{co:SI}, convex restrictions can be obtained by choosing binary matrices $T$ and $R$ as per (\ref{eq:TR<S}) that satisfy the SI condition (\ref{eq:sparsity_invariance}), and requiring that $Y\Gamma \in \text{Sparse}(T)$ and $X\Gamma \in \text{Sparse}(R)$ for any invertible real matrix $\Gamma \in \mathbb{R}^{n \times n}$. We refer the interested reader to \cite{furieri2019separable} for details. 

Based on the discussion above, SI  is a framework-independent notion which deals with sparsity patterns. Specifically, the SI notion translates, separately, to generalizations of QI-based synthesis of sparse dynamic controllers and of block-diagonal quadratic Lyapunov functions for designing sparse static controllers.

\section{Experiments}
\label{se:numerical}
With the goal of providing insight into our proposed method and showing its potential benefits when combined with standard controller design techniques, we continue here our Example~\ref{ex:1} and provide numerical results.


\subsection{Finite-dimensional approximation}
Since the convex programs we have cast are infinite-dimensional, due to the decision variables being transfer matrices whose order is not fixed, it is necessary to resort to finite-dimensional approximations. When using the Youla parametrization in continuous-time, one can adapt the semidefinite programming technique of \cite{alavian2013q} to the $\mathcal{H}_2$ norm by exploiting standard results from \cite{scherer2000efficient,scherer2000linear}; when using the SLP or IOP parametrizations in discrete-time, one can use the corresponding finite impulse response (FIR) approximations of \cite{wang2019system,furieri2019input}. The key common idea behind these approximations is to express each decision variable $\mathbf{U}$, which is a general stable  transfer matrix in continuous-time (resp. discrete-time), in the approximated form 
\begin{equation}
\label{eq:Q_finite_dim}
\mathbf{U}=\sum_{i=0}^NU[i](s+a)^{-i}\,, \quad \left( \text{ resp. }\mathbf{U}=\sum_{i=0}^NU[i]z^{-i}\right)\,,
\end{equation} 
for some $N \in \mathbb{N}$ and $a \in \mathbb{R}$ with $a>0$. The real matrices $U[i]$ for all $i$ become the finitely many real decision variables to optimize over. The approximation \eqref{eq:Q_finite_dim} is based on the well-known idea of Ritz approximations~\cite{boyd1991linear} and we refer the reader to \cite{wang2019system,furieri2019input} for details on SLP and IOP.


\emph{Example~\ref{ex:1} (continued)} We will address the distributed controller design problem formulated in Example~\ref{ex:1} both in discrete- and continuous-time. We have observed in Example~\ref{ex:QI_theorem} that $\text{Sparse}(S)$ is not QI with respect to $\mathbf{G}$. As we have  summarized in Section~\ref{sub:conv_restrictions_QI},  \cite{rotkowitz2012nearest} suggests identifying a binary matrix $T_{\text{QI}} < S$ such that $\text{Sparse}(T_{\text{QI}})$ is QI with respect to $\mathbf{G}$ and $\|S\|_0-\|T_{\text{QI}}\|_0$ is minimized. In this case,  we verify by inspection that $S_2$ in (\ref{eq:S2}) is the only QI sparsity pattern $T_{\text{QI}}$ such that $\|S\|_0-\|T_{\text{QI}}\|_0\leq 2$.  As suggested in \cite{rotkowitz2012nearest}, we can thus substitute the constraint $\mathbf{Y}_Q(\mathbf{X}_Q)^{-1} \in \text{Sparse}(S)$ with $\mathbf{Y}_Q \in \text{Sparse}(S_2)$ and the corresponding convex program is a restriction of $\mathcal{P}_K$. Our goal is to compare tightness of this convex restriction with that of $\mathcal{P}_{S,R^\star_S}$ obtained through SI. 
\subsection{Numerical Results}

 As outlined above, we solved finite-dimensional approximations of the convex restriction proposed in \cite{rotkowitz2012nearest} and of our convex restriction $\mathcal{P}_{S,R^\star_S}$ with $\mathbf{\Gamma}=I_p$ obtained through SI. All the numerical programs were solved with MOSEK \cite{mosek}, called through MATLAB via YALMIP \cite{YALMIP}, on a standard laptop computer.

 \subsubsection{IOP in discrete-time} In our first experiment we considered the discrete-time version of $\mathbf{G}$. Since the approach of \cite{alavian2013q} requires finding an initial stable and stabilizing controller in $\text{Sparse}(S)$ heuristically, which is no trivial task in general, we used the IOP parametrization \cite{furieri2019input} and the discrete-time finite-dimensional approximation \eqref{eq:Q_finite_dim} for all decision variables. Using the notation of \cite{zheng2019equivalence}, where $\mathbf{K}=\mathbf{UY}^{-1}$ and $\mathbf{U}$, $\mathbf{Y}$ are input-output parameters, the closest QI subset approach of \cite{rotkowitz2012nearest} requires $\mathbf{U} \in \text{Sparse}(S_2)$, while our SI approach translates to $\mathbf{U} \in \text{Sparse}(S)$ and $\mathbf{Y} \in \text{Sparse}(R_S^\star)$. Within this setting, no feasible solution could be obtained using the closest QI subset approach; instead, upon convergence over $N$, we obtained a cost of $6.7278$ using the proposed SI approach. To evaluate the suboptimality, we additionally solved for the nearest QI \emph{superset} of $S$ defined as the binary matrix $S_3\geq S$ such that $S_3$ is QI and $\|S_3\|_0-\|S\|_0$ is minimized \cite{rotkowitz2012nearest}; the corresponding optimal cost serves as a lower bound for that of $\mathcal{P}_K$. The QI superset is unique  and is computed with the algorithm (13)-(14) of \cite{rotkowitz2012nearest}. It turns out that $S_3$ is the full lower-triangular matrix. By solving for $S_3$ we obtained the lower bound $6.7268$ upon convergence over $N$, and hence the SI solution has near-optimal performance. 

\subsubsection{Youla in continuous-time} In our second experiment we considered the continuous-time version of $\mathbf{G}$ and used the finite-dimensional approximation technique of \cite{alavian2013q}. A doubly-coprime factorization of $\mathbf{G}$  was computed as per  \cite[Theorem~17]{rotkowitz2006characterization} using the stable and stabilizing controller $\mathbf{K}_{\text{nom}}$ suggested in \cite[Page 1995]{rotkowitz2006characterization}. In (\ref{eq:Q_finite_dim}), we chose $a>0$ and increased the value of $N$ until the improvement on the cost was negligible, thus approaching convergence to the optimal cost of the infinite-dimensional program.  Upon convergence over $N$, the closest QI subset method of \cite{rotkowitz2012nearest} led to a cost of $7.3367$ while the SI method led to a cost of $7.3098$. To evaluate this improvement in performance, we additionally solved for  $S_3$ and obtained a lower bound of $7.2163$. We conclude that our SI solution has a relative improvement over that of \cite{rotkowitz2012nearest} based on QI subsets of at least $\frac{7.3367-7.3098}{7.3367-7.2163}=22.3\%$.

\section{Conclusions}
\label{se:conclusion}
We have proposed the framework of Sparsity Invariance (SI) for convex design of optimal and near-optimal sparse controllers.  One main insight is that the proposed SI approach offers a direct generalization of previous design methods based on the notion of Quadratic Invariance (QI).  Indeed, SI can be directly applied to any systems and sparsity constraints.  The recovered solution is globally optimal when QI holds and performs at least as well as the nearest QI subset when QI does not hold. We have shown the potential benefits of SI over previous methods through examples, and remarked that SI is naturally applicable to sparse static controller design.

Since the condition (\ref{eq:TR<S}) is necessary and sufficient for the SI property (\ref{eq:sparsity_invariance}), our results approach the limits in performance of convex restrictions of the sparsity constrained control problem based on structural conditions for the Youla parameter. This opens up the question of whether  different and more performing design methodologies can be developed for this challenging problem. Another direction for research is to further refine the SI approach, by developing tractable heuristics to optimally design the binary matrices $T$ and $R$ and the parameter $\mathbf{\Gamma}$ simultaneously based on the knowledge of the system $\mathbf{P}$. This could potentially improve upon Algorithm~\ref{alg:procedure}. Finally, it would be relevant to extend the SI idea to the case of delay constraints; in discrete-time, this might be possible by refining the results of \cite{furieri2018robust}. 

\appendix
\section{Appendix}
\subsection{Proof of Theorem~\ref{th:sparsity_invariance}}
\label{appsub:theorem}

\setcounter{lemma}{0}
     \renewcommand{\thelemma}{\Alph{section}\arabic{lemma}}
The proof relies on two Lemmas, whose proofs are reported in Appendix~\ref{appsub:le:grows_fullest} and Appendix~\ref{appsub:le:T}.
\begin{lemma}
\label{le:grows_fullest}
Let $R \in \{0,1\}^{p \times p}$  with $R \geq I_p$. Then, 
\begin{enumerate}
\item For any invertible transfer matrix $\mathbf{X}$ in $\text{Sparse}\left(R\right)$,
\begin{equation*}
\text{Struct}\left(\mathbf{X}^{-1}\right) \leq R^{p-1}\,.
\end{equation*}
\item There exists an invertible transfer matrix $\mathbf{X}\in \text{Sparse}(R)$ such that
\begin{equation*}
\text{Struct}\left(\mathbf{X}^{-1}\right)=R^{p-1}\,.
\end{equation*}
\end{enumerate}
\end{lemma}

\begin{lemma}
\label{le:T}
Let $T \in \{0,1\}^{m \times p}$ and $R \in \{0,1\}^{p \times p}$, and  $\text{Struct}(
\mathbf{W})=R$. Then, there exists $\mathbf{Z} \in \text{Sparse}(T)$ such that
$$\text{Struct}(\mathbf{ZW})=TR\,.$$
\end{lemma}

We are now ready to prove Theorem~\ref{th:sparsity_invariance}.

$1) \Rightarrow 2)$: Let $\mathbf{X} \in \text{Sparse}(R)$ be invertible.  By Lemma~\ref{le:grows_fullest} we know that $\mathbf{X}^{-1} \in \text{Sparse}(R^{p-1})$. Now let $\mathbf{Y} \in \text{Sparse}(T)$. Since $TR^{p-1} \leq S$, we have $\mathbf{YX}^{-1} \in \text{Sparse}(S)$.

$2) \Rightarrow 1)$: We prove by contrapositive. First, suppose that $TR^{p-1} \not \leq S$. By the second statement of  Lemma~\ref{le:grows_fullest} it is possible to select $\mathbf{X}\in \text{Sparse}(R)$ such that $\text{Struct}(\mathbf{X}^{-1})=R^{p-1}$. By the latter and Lemma~\ref{le:T}, we can select $\mathbf{Y} \in \text{Sparse}(T)$ such that $\text{Struct}\left(\mathbf{YX}^{-1}\right)=TR^{p-1}$, or equivalently $\mathbf{YX}^{-1} \not \in \text{Sparse}(S)$. Next, suppose that $T\not \leq S$. Since $R\geq I_p$ by hypothesis, then $TR \not \leq S$ and $TR^{p-1} \not\leq S$. Hence, the same reasoning applies.

\subsection{Proof of Lemma~\ref{le:grows_fullest}}
\label{appsub:le:grows_fullest}
Suppose $\mathbf{X} \in \text{Sparse}(R)$ is invertible. By Cayley-Hamilton's theorem $\sum_{i=0}^{n}\lambda_i\mathbf{X}^i =0$ where $\{\lambda_i\}_{i=0}^{p}$, $\lambda_i \in \mathcal{R}_c$  for every $i=1,\ldots,p$ are the coefficients of the characteristic polynomial of $\mathbf{X}$ and $\lambda_0=\det{\mathbf{X}}\neq 0$. We remark that Cayley-Hamilton is valid over square matrices defined over a commutative ring, such as that of causal transfer functions \cite{picard1995feedback}. By pre-multiplying by $\mathbf{X}^{-1}$ and rearranging the terms:
\begin{equation}
\label{eq:inverse_sum}
\mathbf{X}\mathbf{}^{-1}=-\lambda_0^{-1}(\lambda_1 I_p+\lambda_2\mathbf{X}+\lambda_3\mathbf{X}^2+\cdots+\lambda_p\mathbf{X}^{p-1})\,.
\end{equation}
Since $R\geq I_p$ we have that $R^a \geq R^b$ for every integer $a\geq b$. Hence,  $\lambda_i \mathbf{X}^i \in \text{Sparse}\left(R^{p-1}\right)$ for every $i$ and the first statement follows by (\ref{eq:inverse_sum}).

For the second statement, we iteratively construct $\mathbf{X}$ starting from $\mathbf{X}=I_p$. Let $\mathbf{\alpha} \in \mathcal{R}_c$. Define $\tilde{\mathbf{X}}=\mathbf{X}+\mathbf{\alpha} e_i e_j^\mathsf{T}$.  Let  $\mathbf{X}^{-1}_{:,i} \in \mathcal{R}_c^{p \times 1}$  and  $\mathbf{X}^{-1}_{j,:}\in \mathcal{R}_c^{1 \times p}$ be the $i$-th column  and the $j$-th row of $\mathbf{X}^{-1}$ respectively, and let $\mathbf{X}^{-1}_{ij}$ be the entry $(i,j)$ of $\mathbf{X}^{-1}$. Using the Sherman-Morrison identity \cite{sherman1950adjustment}, if $\tilde{\mathbf{X}}$ is invertible we obtain
\begin{equation}
\label{eq:adding_to_i}
\tilde{\mathbf{X}}^{-1}_{i,:}=\mathbf{X}^{-1}_{i,:}-\frac{\mathbf{\alpha} \mathbf{X}^{-1}_{ii}}{1+\mathbf{\alpha} \mathbf{X}^{-1}_{ji}} \mathbf{X}^{-1}_{j,:}\,.
\end{equation}
 Recall that each entry of a transfer matrix is a transfer function defined over $s=j\omega$. Hence, by the definition of an invertible transfer matrix (see Section~\ref{se:preliminaries}), (\ref{eq:adding_to_i}) holds for almost every $\omega \in \mathbb{R}$. From (\ref{eq:adding_to_i}), for any $i$ and $\mathbf{\alpha} \in \mathcal{R}_c$, if $\mathbf{X}^{-1}_{ii}\neq 0$, then $\tilde{\mathbf{X}}^{-1}_{ii} \neq 0$. It follows that by choosing $\mathbf{\alpha}$ such that
\begin{align}
\label{eq:alpha_condition}
&\mathbf{\alpha} \mathbf{X}^{-1}_{ji} \neq -1 \text{ and }\mathbf{\alpha} \left(\mathbf{X}^{-1}_{ii} \mathbf{X}^{-1}_{jk}-\mathbf{X}^{-1}_{ji}\mathbf{X}^{-1}_{ik}\right) \neq \mathbf{X}^{-1} _{ik}\nonumber\\
&\text{for almost all }\omega \in \mathbb{R}\,,\nonumber\\
& \forall k\text{ subject to }\mathbf{X}^{-1}_{jk}\text{ and }\mathbf{X}^{-1}_{ik} \text{ are not both null}\,,
\end{align}
we obtain that
\begin{align}
\label{eq:augment_structure}
&\text{Struct}\left(\tilde{\mathbf{X}}^{-1}_{i,:}\right)=\text{Struct}\left(\mathbf{X}^{-1}_{i,:}\right)+\text{Struct}\left( \mathbf{X}^{-1}_{j,:}\right)\,,\\
&\text{for almost all }\omega \in \mathbb{R}\nonumber\,.
\end{align}
 The condition (\ref{eq:alpha_condition}) is derived by setting the right hand side of (\ref{eq:adding_to_i}) to be different from $0$ for every $k$ such that  $\mathbf{X}^{-1}_{ik}$ and $\mathbf{X}^{-1}_{jk}$ are not both null for every $\omega \in \mathbb{R}$. Observe that $\mathbf{\alpha}$ as per (\ref{eq:alpha_condition}) always exists, because there is no $k$ such that $\mathbf{X}^{-1}_{ik}$ and $\mathbf{X}^{-1}_{jk}$ are both null for every $\omega \in \mathbb{R}$, and hence $\mathbf{\alpha} \left(\mathbf{X}^{-1}_{ii} \mathbf{X}^{-1}_{jk}-\mathbf{X}^{-1}_{ji}\mathbf{X}^{-1}_{ik}\right) \neq \mathbf{X}^{-1} _{ik}$ always admits a solution in $\alpha \in \mathcal{R}_c$. The structural augmentation (\ref{eq:augment_structure}) is exploited in the algorithm below.
	\begin{algorithm}[H]
\label{alg:buildXinv}
  \begin{algorithmic}[1]
  \State Set $\mathbf{X}=I_p$
  \Repeat
  \Comment{max. $(|R|-p)(p-1)$ iterations}
  \For{ each $(i,j)$ such that $i\neq j$ and $R_{ij}=1$ }
  	\State Choose $\mathbf{\alpha}$  according to (\ref{eq:alpha_condition})
  	\State $\mathbf{X}\leftarrow \mathbf{X}+\mathbf{\alpha} e_ie_j^\mathsf{T}$
  \EndFor
  \Until{$\text{Struct}(\mathbf{X}^{-1})=R^{p-1}$}
  \State Return $\mathbf{X}$
  \end{algorithmic}
\end{algorithm}
The algorithm  returns a matrix $\mathbf{X}$ such that $\text{Struct}(\mathbf{X}^{-1}) =R^{p-1}$. Specifically, by exploiting (\ref{eq:augment_structure}) we obtain that $\text{Struct}(\mathbf{X}^{-1}) \geq R^{s}$ at the end of the $s$-th iteration of the ``repeat-until'' cycle.

\subsection{Proof of Lemma~\ref{le:T}}
\label{appsub:le:T}

Let $\mathbf{Z}$ be any transfer matrix in $\text{Sparse}(T)$. Assume that $\text{Struct}(\mathbf{ZW})<TR$. Then, for some $(i,j,k)$ we have that $\mathbf{ZW}_{ij}=0$ and $T_{ik}=R_{kj}=1$. We know by hypothesis that  $\mathbf{W}_{kj}\neq 0$. Since $\sum_{l=1}^p\mathbf{Z}_{il}\mathbf{W}_{lj}=0$, it is sufficient to update $\mathbf{Z}_{ik}$ with $\mathbf{Z}_{ik}+\mathbf{\alpha}$ for any $\mathbf{\alpha} \neq 0$ in $\mathcal{R}_c$ to guarantee that $\mathbf{ZW}_{ij} \neq 0$. Furthermore, by choosing $\mathbf{\alpha} \neq -\frac{\mathbf{ZW}_{it}}{\mathbf{W}_{kt}}$ for all $t$ such that $\mathbf{ZW}_{it}\neq 0$, we avoid that adding $\mathbf{\alpha}$ to $\mathbf{Z}_{ik}$ brings $\mathbf{ZW}_{it}$ to $0$ when $\mathbf{ZW}_{it}\neq 0$. Hence, it is always possible to choose $k$ and $\mathbf{\alpha}$ such that $\mathbf{ZW}+\mathbf{\alpha} e_ie_k^\mathsf{T}>\mathbf{ZW}$ and $\mathbf{Z} \in \text{Sparse}(T)$. By iterating the procedure for all $(i,j)$ such that $\text{Struct}(\mathbf{ZW})_{ij}<TR_{ij}$, we converge to $\text{Struct}(\mathbf{ZW})=TR$.

	\bibliographystyle{IEEEtran}

	\bibliography{IEEEabrv,references2}
	
\end{document}

%% file: figs/f14_2.tex
\setlength{\unitlength}{0.008in}
\begin{picture}(243,135)(140,410)
\thicklines
\put(187,440){\vector(1, 0){ 43}}
\put(180,510){\vector( 0, -1){ 65}}
\put(128,440){\vector( 1, 0){ 48}}
\put(140,445){\makebox(0,0)[lb]{$v_1$}}
\put(181,440){\circle{10}}
\put(180,510){\line( 1, 0){ 40}}

\put(340,440){\circle{10}}
\put(340,510){\vector( -1, 0){ 40}}
\put(340,510){\line( 0,-1){ 65}}
\put(290,440){\vector(1, 0){ 45}}
\put(393,440){\vector(-1, 0){ 48}}
\put(365,445){\makebox(0,0)[lb]{$v_2$}}

\put(220,545){\vector( -1, 0){ 50}}
\put(350,545){\vector( -1, 0){ 50}}
\put(230,420){\framebox(60,40){}}
\put(220,500){\framebox(80,60){}}
\put(165,480){\makebox(0,0)[lb]{$y$}}
\put(345,480){\makebox(0,0)[lb]{$u$}}
\put(320,552){\makebox(0,0)[lb]{$w$}}
\put(190,552){\makebox(0,0)[lb]{$z$}}
\put(252,435){\makebox(0,0)[lb]{$\mathbf{K}$}}
\put(225,510){\makebox(0,0)[lb]{$\begin{matrix}\mathbf{P}_{11}&\mathbf{P}_{12}\\\mathbf{P}_{21}&\mathbf{G}\end{matrix}$}}
\end{picture}

%% file: figs/fig_infostructure.tex
\setlength{\unitlength}{0.008in}
\begin{picture}(420,160)(0,0)

\put(0,20){\circle{22}}
\put(-8,14){\makebox(0,0)[lb]{$y_5$}}

\put(0,55){\circle{22}}
\put(-8,49){\makebox(0,0)[lb]{$y_4$}}

\put(0,90){\circle{22}}
\put(-8,84){\makebox(0,0)[lb]{$y_3$}}

\put(0,125){\circle{22}}
\put(-8,119){\makebox(0,0)[lb]{$y_2$}}

\put(0,160){\circle{22}}
\put(-8,154){\makebox(0,0)[lb]{$y_1$}}

\put(47,10){\framebox(33,20){}}
\put(48,11){\makebox(0,0)[lb]{\small $\mathbf{K}_{5,:}$}}

\put(47,45){\framebox(33,20){}}
\put(48,46){\makebox(0,0)[lb]{\small $\mathbf{K}_{4,:}$}}

\put(47,80){\framebox(33,20){}}
\put(48,81){\makebox(0,0)[lb]{\small $\mathbf{K}_{3,:}$}}

\put(47,115){\framebox(33,20){}}
\put(48,116){\makebox(0,0)[lb]{\small $\mathbf{K}_{2,:}$}}

\put(47,150){\framebox(33,20){}}
\put(48,151){\makebox(0,0)[lb]{\small $\mathbf{K}_{1,:}$}}

\put(11,160){\vector( 1, 0){35}}
\put(11,160){\vector( 1, -1){35}}

\put(11,125){\vector( 1, 0){35}}
\put(11,125){\vector( 1, -1){35}}
\put(11,125){\vector( 1, -2){35}}
\put(11,125){\vector( 1, -3){35}}

\put(11,90){\vector( 1, 0){35}}
\put(11,90){\vector( 1, -1){35}}
\put(11,90){\vector( 1, -2){35}}

\put(11,55){\vector( 1, 0){35}}
\put(11,55){\vector( 1, -1){35}}

\put(11,20){\vector( 1, 0){35}}

\put(80,160){\vector( 1, 0){35}}
\put(80,125){\vector( 1, 0){35}}
\put(80,90){\vector( 1, 0){35}}
\put(80,55){\vector( 1, 0){35}}
\put(80,20){\vector( 1, 0){35}}

\put(127,20){\circle{22}}
\put(118,14){\makebox(0,0)[lb]{$u_5$}}

\put(127,55){\circle{22}}
\put(118,49){\makebox(0,0)[lb]{$u_4$}}

\put(127,90){\circle{22}}
\put(118,84){\makebox(0,0)[lb]{$u_3$}}

\put(127,125){\circle{22}}
\put(118,119){\makebox(0,0)[lb]{$u_2$}}

\put(127,160){\circle{22}}
\put(118,154){\makebox(0,0)[lb]{$u_1$}}

\put(180,20){\circle{22}}
\put(172,14){\makebox(0,0)[lb]{$y_5$}}

\put(180,55){\circle{22}}
\put(172,49){\makebox(0,0)[lb]{$y_4$}}

\put(180,90){\circle{22}}
\put(172,84){\makebox(0,0)[lb]{$y_3$}}

\put(180,125){\circle{22}}
\put(172,119){\makebox(0,0)[lb]{$y_2$}}

\put(180,160){\circle{22}}
\put(172,154){\makebox(0,0)[lb]{$y_1$}}

\put(227,8){\framebox(33,25){}}
\put(228,9){\makebox(0,0)[lb]{\small $\mathbf{X}^{-1}_{5,:}$}}

\put(227,43){\framebox(33,25){}}
\put(228,44){\makebox(0,0)[lb]{\small $\mathbf{X}^{-1}_{4,:}$}}

\put(228,78){\framebox(33,25){}}
\put(228,79){\makebox(0,0)[lb]{\small $\mathbf{X}^{-1}_{3,:}$}}

\put(227,113){\framebox(33,25){}}
\put(228,114){\makebox(0,0)[lb]{\small $\mathbf{X}^{-1}_{2,:}$}}

\put(227,148){\framebox(33,25){}}
\put(228,149){\makebox(0,0)[lb]{\small $\mathbf{X}^{-1}_{1,:}$}}

\put(191,160){\vector( 1, 0){35}}

\put(191,125){\vector( 1, 0){35}}
\put(191,125){\vector( 1, -1){35}}
\put(191,125){\vector( 1, -2){35}}
\put(191,125){\vector( 1, -3){35}}

\put(191,90){\vector( 1, 0){35}}
\put(191,90){\vector( 1, -1){35}}
\put(191,90){\vector( 1, -2){35}}

\put(191,55){\vector( 1, 0){35}}
\put(191,55){\vector( 1, -1){35}}


\put(191,20){\vector( 1, 0){35}}

\put(260,160){\vector( 1, 0){35}}
\put(260,160){\vector( 1, -1){35}}

\put(260,125){\vector( 1, 0){35}}
\put(260,125){\vector( 1, -1){35}}
\put(260,125){\vector( 1, -2){35}}
\put(260,125){\vector( 1, -3){35}}

\put(260,90){\vector( 1, 0){35}}
\put(260,90){\vector( 1, -1){35}}
\put(260,90){\vector( 1, -2){35}}

\put(260,55){\vector( 1, 0){35}}
\put(260,55){\vector( 1, -1){35}}

\put(260,20){\vector( 1, 0){35}}

\put(295,10){\framebox(33,20){}}
\put(296,11){\makebox(0,0)[lb]{\small $\mathbf{Y}_{5,:}$}}

\put(295,45){\framebox(33,20){}}
\put(296,46){\makebox(0,0)[lb]{\small $\mathbf{Y}_{4,:}$}}

\put(295,80){\framebox(33,20){}}
\put(296,81){\makebox(0,0)[lb]{\small $\mathbf{Y}_{3,:}$}}

\put(295,115){\framebox(33,20){}}
\put(296,116){\makebox(0,0)[lb]{\small $\mathbf{Y}_{2,:}$}}

\put(295,150){\framebox(33,20){}}
\put(296,151){\makebox(0,0)[lb]{\small $\mathbf{Y}_{1,:}$}}

\put(328,160){\vector( 1, 0){35}}
\put(328,125){\vector( 1, 0){35}}
\put(328,90){\vector( 1, 0){35}}
\put(328,55){\vector( 1, 0){35}}
\put(328,20){\vector( 1, 0){35}}

\put(375,20){\circle{22}}
\put(367,14){\makebox(0,0)[lb]{$u_5$}}

\put(375,55){\circle{22}}
\put(367,49){\makebox(0,0)[lb]{$u_4$}}

\put(375,90){\circle{22}}
\put(367,84){\makebox(0,0)[lb]{$u_3$}}

\put(375,125){\circle{22}}
\put(367,119){\makebox(0,0)[lb]{$u_2$}}

\put(375,160){\circle{22}}
\put(367,154){\makebox(0,0)[lb]{$u_1$}}


\end{picture}